\pdfoutput=1
\documentclass[11pt,a4paper]{article}

\usepackage{fullpage}
\usepackage[english]{babel}
\usepackage[inline]{enumitem}
\usepackage[fleqn]{amsmath}
\usepackage{amsthm}
\usepackage{amssymb}
\usepackage{mathtools}
\usepackage[style=alphabetic,citestyle=alphabetic,maxbibnames=9]{biblatex}
\usepackage[hidelinks,pdfusetitle,pdfencoding=auto,psdextra,bookmarksopen]{hyperref}
\usepackage[textsize=tiny,backgroundcolor=white,linecolor=black]{todonotes}
\usepackage{tikz}
\usepackage[normalem]{ulem}

\usepackage{amsthm}
\newtheorem{theorem}{Theorem}

\newtheorem{lemma}[theorem]{Lemma}
\newtheorem{claim}[theorem]{Claim}
\newtheorem{corollary}[theorem]{Corollary}
\newtheorem{observation}[theorem]{Observation}

\setlist[enumerate]{label=(\roman*), itemsep=0pt}
\setlist[itemize]{itemsep=0pt}
\setlist{beginpenalty=10000, midpenalty=10000}
\def\P{\ensuremath{\mathrm{P}}}
\def\NP{\ensuremath{\mathrm{NP}}}
\def\NE{\ensuremath{\mathrm{NE}}}
\def\NEE{\ensuremath{\mathrm{NEE}}}
\def\FP{\ensuremath{\mathrm{FP}}}
\def\UP{\ensuremath{\mathrm{UP}}}
\def\DisjNP{\ensuremath{\mathrm{DisjNP}}}
\def\DisjCoNP{\ensuremath{\mathrm{DisjCoNP}}}
\def\DisjUP{\ensuremath{\mathrm{DisjUP}}}
\def\DisjCoUP{\ensuremath{\mathrm{DisjCoUP}}}
\def\coNP{\ensuremath{\mathrm{coNP}}}
\def\coNE{\ensuremath{\mathrm{coNE}}}
\def\coNEE{\ensuremath{\mathrm{coNEE}}}
\def\coUP{\ensuremath{\mathrm{coUP}}}
\def\NPcoNP{\ensuremath{\mathrm{NP}\cap\mathrm{coNP}}}
\def\TFNP{\ensuremath{\mathrm{TFNP}}}
\def\TALLY{\ensuremath{\mathrm{TALLY}}}
\def\NPMV{\ensuremath{\mathrm{NPMV}}}
\def\NPMVt{\ensuremath{\mathrm{NPMV_t}}}
\def\NPSV{\ensuremath{\mathrm{NPSV}}}
\def\NPSVt{\ensuremath{\mathrm{NPSV_t}}}
\def\NPbV{\ensuremath{\mathrm{NPbV}}}
\def\NPbVt{\ensuremath{\mathrm{NPbV_t}}}
\def\NPkV{\ensuremath{\mathrm{NP}k\mathrm{V}}}
\def\NPkVt{\ensuremath{\mathrm{NP}k\mathrm{V_t}}}
\def\TAUT{\ensuremath{\mathrm{TAUT}}}
\def\SAT{\ensuremath{\mathrm{SAT}}}
\def\PF{\ensuremath{\mathrm{PF}}}
\DeclareMathOperator{\dom}{dom}
\DeclareMathOperator{\img}{img}
\DeclareMathOperator{\supp}{supp}
\def\hUP{\ensuremath{\mathsf{UP}}}
\def\hDisjNP{\ensuremath{\mathsf{DisjNP}}}
\def\hDisjCoNP{\ensuremath{\mathsf{DisjCoNP}}}
\def\hNPcoNP{\ensuremath{\mathsf{NP}{}\cap{}\mathsf{coNP}}}
\def\hCON{\ensuremath{\mathsf{CON}}}
\def\hSAT{\ensuremath{\mathsf{SAT}}}
\def\hTFNP{\ensuremath{\mathsf{TFNP}}}
\def\leqmpp{\ensuremath{\leq_\mathrm{m}^\mathrm{pp}}}
\def\leqmp{\ensuremath{\leq_\mathrm{m}^\mathrm{p}}}

\newcommand{\impl}{\,\mathop{\Rightarrow}\,}
\renewcommand{\iff}{\,\mathop{\Longleftrightarrow}\,}

\def\sqsubsetneq{\mathrel{\sqsubseteq\kern-0.92em\raise-0.15em\hbox{\rotatebox{313}{\scalebox{1.1}[0.75]{\(\shortmid\)}}}\scalebox{0.3}[1]{\ }}}
\def\sqsupsetneq{\mathrel{\sqsupseteq\kern-0.92em\raise-0.15em\hbox{\rotatebox{313}{\scalebox{1.1}[0.75]{\(\shortmid\)}}}\scalebox{0.3}[1]{\ }}}

\title{Oracle with $\mathrm{P=NP\cap coNP}$, but no Many-One Completeness\\ in
UP, DisjNP, and DisjCoNP}

\author{Anton Ehrmanntraut \and Fabian Egidy \and Christian Glaßer}
\date{Julius-Maximilians-Universität Würzburg\\\small\{\texttt{anton.ehrmanntraut}, \texttt{fabian.egidy}, \texttt{christian.glasser}\}\texttt{@uni-wuerzburg.de}\\[2ex]\today}

\begin{document}
\frenchspacing
\maketitle
\begin{abstract}
We construct an oracle relative to which $\P = \NP \cap \coNP$,
but there are no many-one complete sets in $\UP$,
no many-one complete disjoint $\NP$-pairs, and
no many-one complete disjoint $\coNP$-pairs.

This contributes to a research program initiated by Pudlák~\cite{pud17},
which studies incompleteness in the finite domain and which
mentions the construction of such oracles as open problem.
The oracle shows that $\hNPcoNP$ is indispensable
in the list of hypotheses studied by Pudlák.
Hence one should consider stronger hypotheses,
in order to find a universal one.
\end{abstract}

\section{Introduction}

Questions of the existence of complete sets in promise classes have a long history.
They turned out to be difficult and remained open.
Consider the following examples, where the questions are expressed as hypotheses.
\begin{eqnarray*}
    \hNPcoNP: && \mbox{$\NP \cap \coNP$ does not contain many-one complete problems \ \cite{kan79}\nocite{sip82}} \\
    \hUP: && \mbox{$\UP$ does not contain many-one complete problems \ \cite{hh88}}\\
    \hCON: && \mbox{p-optimal proof systems for $\TAUT$ do not exist \ \cite{kp89}}\\
    \hSAT: && \mbox{p-optimal proof systems for $\SAT$ do not exist \ \cite{ffnr03}}\\
    \hTFNP: && \mbox{$\TFNP$ does not contain many-one complete problems \ \cite{mp91}}\\
    \hDisjNP: && \mbox{$\DisjNP$ does not contain many-one complete pairs \ \cite{raz94}}\\
    \hDisjCoNP: && \mbox{$\DisjCoNP$ does not contain many-one complete pairs \ \cite{mes00,pud14}}
\end{eqnarray*}

So far, the following implications are known:
$\hDisjNP \impl \hCON$ \cite{raz94},
$\hUP \impl \hCON$ \cite{kmt03},
$\hDisjCoNP \impl \hTFNP$ \cite{pud17},
$\hTFNP \impl \hSAT$ \cite{bkm09,pud17}, and
$\hNPcoNP \impl \hCON \vee \hSAT$ \cite{kmt03}.
This raises the question of whether further implications are provable with the currently available means.
Thanks to a work by Pudlák~\cite{pud17}, this question recently gained momentum.
In fact, Pudlák's interest goes beyond:
He initiated a research program to find a general principle from which the remaining hypotheses follow as special cases.
This is motivated by the study of incompleteness in the finite domain,
since these hypotheses can either be expressed as the non-existence of complete elements in promise classes
or as statements about the unprovability of sentences of some specific form in weak theories.

Pudlák~\cite{pud17} states as open problem to construct oracles that show that the relativized conjectures are different or show that they are equivalent.
Such oracles have been constructed by
Verbitskii~\cite{ver91},
Glaßer et al.~\cite{gssz04},
Khaniki~\cite{kha19},
Dose~\cite{dos20a,dos20b,dos20c}, and
Dose and Glaßer~\cite{dg20}.
The restriction to relativizable proofs arises from the following idea:
We consider the mentioned hypotheses as conjectures,
hence we expect that they are equivalent.
In this situation we are not primarily concerned
with the question of whether two hypotheses are equivalent,
but rather whether their equivalence can be
{\em recognized} with the currently available means.
An accepted formalization of this is the notion of relativizable proofs.

\paragraph{Our Contribution}
We contribute to Pudlák's program with the construction of an oracle relative to which the following holds: $\hUP$, $\hDisjNP$, $\hDisjCoNP$, but $\P = \NP \cap \coNP$.
Hence there is no relativizable proof for $\hNPcoNP$, even if we simultaneously assume all remaining hypotheses we mentioned so far.
This demonstrates that $\hNPcoNP$ is indispensable in the list of currently viewed hypotheses
and suggests to broaden the focus and include stronger statements.

Pudlák~\cite{pud17} ranks $\hNPcoNP$ as a plausible conjecture that is apparently incomparable with $\hCON$ and $\hTFNP$.
Our oracle supports this estimation, as it rules out relativizable proofs for ``$\hCON \impl \hNPcoNP$'' and ``$\hTFNP \impl \hNPcoNP$.''
By Dose~\cite{dos20a,dos20b}, the same holds for the converse implications.
Overall, we recognize a strong independence between $\hNPcoNP$ and all remaining hypotheses:
\begin{enumerate}
    \item There does not exist a relativizable proof for $\hNPcoNP$, even if we simultaneously assume all remaining hypotheses.
    \item There exists a relativizable proof for the implication $\hNPcoNP \impl \hCON \vee \hSAT$~\cite{kmt03}.
    But there does not exist a relativizable proof showing that $\hNPcoNP$ implies one of the remaining hypotheses~\cite{dos20a,dos20b}.
\end{enumerate}

Our oracle combines several separations with the collapse $\P = \NP \cap \coNP$.
This leads to conclusions on the independence of the statement $\P \neq \NP \cap \coNP$ from typical assumptions.
For instance, the oracle shows that $\P \neq \NP \cap \coNP$ cannot be proved by relativizing means,
even under the strong but likely assumption $\hUP \wedge \hDisjNP \wedge \hDisjCoNP$.

Further characteristics of our oracle are, for example,
$\NE \neq \coNE$, $\NPMV \not\subseteq_c \NPSV$, and the shrinking and separation properties do not hold for $\NP$ and $\coNP$.
Corollary~\ref{cor:detailed_properties} presents a list of additional properties.

\paragraph{Open questions.}
Currently, for almost every pair $\mathsf A$, $\mathsf B$ of the discussed hypotheses, we either know a relativizable proof for the implication $\mathsf A\Rightarrow \mathsf B$, or we know an oracle relative to which $\mathsf A \land \neg\mathsf B$. (See also Figure~\ref{fig:diagram}.)
Only three cases are left:
\begin{enumerate}
\item $\hUP \mathop{\Rightarrow}\limits^{?} \hDisjNP$
\item $\hTFNP \mathop{\Rightarrow}\limits^{?} \hDisjCoNP$
\item $\hSAT \mathop{\Rightarrow}\limits^{?} \hTFNP$
\end{enumerate}
This leads to the following task for future research:
Prove these implications or construct oracles relative to which they do not hold.

\begin{figure}[tb]
    \centering
    \begin{tikzpicture}[node/.style={anchor=base},tips=proper]
        \node (pneqnp) at (0,0) {$\mathrm{P\neq NP}$};
        \node [above=1 of pneqnp] (conorsat) {$\mathsf{CON} \lor \mathsf{SAT}$};
        \node [above left=0.2 and 2 of conorsat] (con) {$\mathsf{CON}$};
        \node [above right=0.2 and 2 of conorsat] (sat) {$\mathsf{SAT}$};
        \node [above left=3.0cm and 1.7cm of con,anchor=base] (disjnp) {$\mathsf{DisjNP}$};
        \node [above right=3.0cm and 1.7cm of sat,anchor=base] (disjconp) {$\mathsf{DisjCoNP}$};
        \node [above right=1.4cm and 0.7cm of sat,anchor=base] (tfnp) {$\mathsf{TFNP}$};
        \node [above =1.8cm of conorsat] (npiconp) {$\mathsf{NP{\cap}coNP}$};
        \node [above=1.4cm of con] (up) {$\mathsf{UP}$};
        \node [above=5.5 of npiconp] (dgh) {$\mathsf{DisjNP\land NP{\cap}coNP}$};
        \node [above=1 of dgh] (titush) {$\mathsf{DisjNP\land UP\land NP{\cap}coNP}$};
        \node [right=1 of dgh] (ourh) {$\mathsf{DisjNP\land UP\land DisjCoNP}$};

        \draw [-latex,ultra thick] (con) -- (conorsat);
        \draw [-latex,ultra thick] (sat) -- (conorsat);
        \draw [-latex,ultra thick] (conorsat) -- (pneqnp);
        \draw [-latex,ultra thick] (disjnp) -- (con);
        \draw [-latex,ultra thick] (up) -- (con);
        \draw [-latex,ultra thick] (disjconp) -- (tfnp);
        \draw [-latex,ultra thick] (tfnp)-- (sat);
        \draw [-latex,ultra thick] (npiconp) -- (conorsat);
        \draw [-latex] (ourh) -- (disjnp);
        \draw [-latex] (ourh) -- (up);
        \draw [-latex] (ourh) -- (disjconp);
        \draw [-latex] (titush) -- (dgh);
        \draw [-latex] (titush.190) -- (up);

        \draw [-latex,dashed] (titush.-10) edge node [sloped,pos=0.6,above] {\scriptsize\cite{dos20b}} (sat.north);
        \draw [-latex] (dgh) -- (disjnp);
        \draw [-latex] (dgh) -- (npiconp);
        \draw [dashed, -latex] (dgh) edge  node [sloped,pos=0.6,below] {\scriptsize\cite{dg20}} (up);
        \draw [dashed, -latex] (npiconp) edge  node [sloped,midway,above] {\scriptsize\cite{dos20a}} (con);
        \draw [dashed, -latex,bend left=15] (disjconp) edge  node [sloped,pos=0.4,above] {\scriptsize\cite{kha19}} (con);
        \draw [dashed, -latex] (ourh) edge  node [sloped,pos=0.7,above] {\scriptsize Thm.~\ref{thm:result}} (npiconp);
        \draw [dashed, -latex, bend left] (con) edge  node [sloped,pos=0.5,below] {\scriptsize\cite{gssz04}} (disjnp);
        \draw [dashed, -latex, bend left=50] (pneqnp.west) edge  node [sloped,pos=0.5,below] {\scriptsize\cite{dos20c}} (conorsat.west);
    \end{tikzpicture}
    \caption{Solid arrows mean implications. All implications occurring in the figure have relativizable proofs.
        (The only nontrivial ones are $\hDisjNP\Rightarrow\hCON$~\cite{raz94}, $\hUP\Rightarrow\hCON$~\cite[Cor.~4.1]{kmt03}, $\hDisjCoNP\Rightarrow\hTFNP\Rightarrow\hSAT$~\cite[Prop.~5.6]{pud17}\cite[Thm.~25]{bkm09}\cite[Prop.~5.10]{pud17}.) Implications between the conjectures originally considered by Pudlák (i.e., not the conjunctions) are highlighted bold. A dashed arrow from one conjecture $\mathsf A$ to another conjecture $\mathsf B$ means that there is an oracle $X$ against the implication $\mathsf A\Rightarrow\mathsf B$, meaning that $\mathsf A \land \neg\mathsf B$ holds relative to $X$.
    }
    \label{fig:diagram}
\end{figure}

\paragraph{Background on connections between promise classes and proof systems.} 
We are mainly interested in the following well-studied promise classes: The class of disjoint $\NP$-pairs $\DisjNP$~\cite{sel88, gs88}, the class of disjoint $\coNP$-pairs $\DisjCoNP$~\cite{ffnr96, ffnr03}, the class of sets accepted by nondeterministic polynomial-time machines with at most one accepting computation path $\UP$~\cite{val76}, the class $\NPcoNP$~\cite{edm65}, and the class of all total polynomial search problems $\TFNP$~\cite{mp91}. Furthermore, we are interested in proof systems defined by Cook and Reckhow~\cite{cr79}, especially proof systems for the set of satisfiable formulas $\SAT$ and proof systems for the set of tautologies $\TAUT$, also called propositional proof systems. 

The connections between propositional proof systems and promise classes have been studied intensively. Krajícek and Pudlák~\cite{kp89} linked propositional proof systems (and thus the hypothesis $\hCON$) to standard complexity classes by proving that $\mathrm{NE} = \coNE$ implies the existence of optimal propositional proof systems and $\mathrm{E} = \NE$ implies the existence of $p$-optimal propositional proof systems. These results were subsequently improved by Köbler, Messner, and Torán~\cite{kmt03}. 

Glaßer, Selman, and Sengupta~\cite{gss05} give several characterizations of $\hDisjNP$. Some characterizations use different notions of reducibility while others use the existence of $\leqmp$-complete functions in $\NPSV$ and the uniform enumerability of disjoint $\NP$-pairs. Glaßer, Selman, and Zhang~\cite{gsz07, gsz09} connect propositional proof systems to disjoint $\NP$-pairs. They prove that the degree structure of $\DisjNP$ and of all canonical disjoint pairs of propositional proof systems is the same. Beyersdorff~\cite{bey04, bey06, bey07, bey10} and Beyersdorff and Sadowkski~\cite{bs11} investigate further connections between disjoint $\NP$-pairs and propositional proof systems. 

Pudlák~\cite{pud96, pud03, pud17} draws connections between the finite consistency problem, proof systems, and promise classes like $\DisjNP$ and $\TFNP$. Moreover, he asks for oracles that separate hypotheses regarding proof systems and promise classes. Several oracles have been constructed since Pudlák formulated his research questions.  Concerning the listed hypotheses, Figure~\ref{fig:diagram} summarizes all known (relativizing) implications and implications that do not hold relative to some oracle.

The paper is organized as follows:  Section~\ref{sec:prelim} defines the complexity classes mentioned above and presents our notations.
Section~\ref{sec:construction} contains the oracle construction: the first part defines the construction, the second part proves that it is well-defined, and the last part shows the claimed properties.

\section{Preliminaries}\label{sec:prelim}

Throughout this paper, let $\Sigma$ be the alphabet $\{0,1\}$.
The set $\Sigma^*$ denotes the set of finite words over $\Sigma$.
The set $\Sigma^\omega$ denotes the set of $\omega$-infinite words, i.e., the $\omega$-infinite sequences of characters from $\Sigma$.
Let $\Sigma^{\leq n} \coloneqq \{w\in\Sigma^* \mid |w|\leq n\}$.
For word $w\in\Sigma^*\cup\Sigma^\omega$, we denote with $w(i)$ the $i$-th character of $w$ for $0\leq i<|w|\leq\omega$.
We write $v\sqsubseteq w$ when $v$ is a prefix of $w$, that is, $|v|\leq|w|$ and $v(i)=w(i)$ for all $0\leq i < |v|$.
Accordingly, $v\sqsubsetneq w$ when $v\sqsubseteq w$ and $v\neq w$.
The empty word is denoted by $\varepsilon$.
For a finite set $A\subseteq\Sigma^*$, we define $\ell(A)\coloneqq\sum_{w\in A} |w|$.

Let $\mathbb N$ denote the set of non-negative integers, and $\mathbb N^+$ the set of positive integers.
We say that two sets $X$ and $Y$ \emph{agree on set $Z$} when $z\in X$ if and only if $z\in Y$ for all $z\in Z$.

The finite words $\Sigma^*$ can be linearly ordered by their quasi-lexicographic (i.e., “shortlex”) order $\prec_\mathrm{lex}$, uniquely defined by requiring $0\prec_\mathrm{lex} 1$.
Under this definition, there is a unique  order-isomorphism between $(\Sigma^*, \prec_\mathrm{lex})$ and $(\mathbb N, <)$, which induces a polynomial-time computable, polynomial-time invertible bijection between $\Sigma^*$ and $\mathbb N$. 
Hence, we can transfer notations, relations, and operations for $\Sigma^*$ to $\mathbb N$ and vice versa.
In particular, $|n|$ denotes the length of word represented by $n\in \mathbb N$.
By definition of $\prec_\mathrm{lex}$, whenever $a\leq b$, then $|a|\leq |b|$.
We eliminate the ambiguity of the expressions $0^i$ and $1^i$ by always interpreting them over $\Sigma^i$.
Moreover, $\leq$ denotes both the less-than-or-equal relation for natural numbers and the quasi-lexicographic order for finite words.
From the properties of order-isomorphism, this is compatible with above identification of words and numbers.
Similarly, we define the \emph{parity} of a word $w\in\Sigma^*$ as the parity of the natural number represented by $w$.
Note that the word $\varepsilon\in\Sigma^*$ represents $0\in\mathbb N$, the word $0\in\Sigma^*$ represents $1\in\mathbb N$, the word $1\in\Sigma^*$ represents $2\in\mathbb N$ and so on.
Hence, all words in $\Sigma^*0$ are odd and all words in $\Sigma^*1$ are even.

For (both directed and undirected) graphs $G$, we denote with $V(G)$ the vertex set of $G$, with $E(G)$ the edge set of $G$.
For a vertex $v\in V(G)$, we denote with $R_G(v)$ the set of vertices that are reachable from $v$, and with $N_G(v)$ (resp., $N_G^+(v)$) the vertices adjacent to $v$ in undirected $G$ (resp., direct successors of $v$ in directed $G$).
Similarly, for a subset $U\subseteq V(G)$, we define $R_G(U)=\bigcup_{v\in U} R_G(v)$, and $N_G(U), N^+_G(U)$ respectively.
Observe that always $v\in R_G(v)$ and $U\subseteq R_G(U)$.
For a directed acyclic graph $G$, we define the \emph{height of vertex $x\in V(G)$} as the length of the longest path in $G$ to a “sink”, i.e., to a vertex with no outgoing edges.

We understand $\P$ (resp., $\NP$) as the usual complexity class of languages decidable by a deterministic (resp., nondeterministic) polynomial-time Turing machine. The class $\FP$ refers to the class of total functions that can be computed by a deterministic polynomial-time Turing transducer \cite{pap81}.
Valiant \cite{val76} defined $\UP$ as the set of all languages that can be recognized by a nondeterministic polynomial-time machine that, on every input, accepts on at most one computation path. We use the definitions from Köbler, Messner, and Torán \cite{kmt03} for the nondeterministic exponential and nondeterministic double exponential time classes $\NE \coloneqq \mathrm{NTIME}\left(2^{O(n)}\right)$ and $\NEE \coloneqq \mathrm{NTIME}\left(2^{O(2^n)}\right)$. Let $\TALLY$ be the class of tally sets, that is, sets $A\subseteq \{ 0^n \mid n \geq 0 \}$. For a complexity class $\mathcal{C}$ we define $\mathrm{co}\mathcal{C} \coloneqq \{\overline{A} \mid A \in \mathcal{C}\}$ as the complementary complexity class of $\mathcal{C}$.

Between sets of words, we employ the usual polynomial-time many-one reducibility: We say that $A$ is \emph{polynomial-time many-one reducible to} $B$, denoted by $A \leqmp B$, if there exists a function $f \in \FP$ such that $x \in A \Leftrightarrow f(x) \in B$ for all $x \in \Sigma ^*$.
For some complexity class $\mathcal C$, we say that $B$ is \emph{$\leqmp$-hard for $\mathcal C$} when $A\leqmp B$ for any $A\in\mathcal C$.
If, additionally, $B\in\mathcal C$, we say that $B$ is \emph{$\leqmp$-complete for $\mathcal C$}.

A \emph{disjoint $\NP$-pair} is a pair $(A,B)$ of disjoint sets in $\NP$. Selman~\cite{sel88} and Grollmann and Selman~\cite{gs88} defined the class $\DisjNP$ as the set of disjoint $\NP$-pairs. The classes $\DisjCoNP$~\cite{ffnr96, ffnr03}, $\DisjUP$, and $\DisjCoUP$ are defined similarly.
Between two pairs, we employ the following related notion of reducibility~\cite{raz94}: Let $(A,B)$ and $(C,D)$ be two disjoint pairs. We say that $(A,B)$ is \emph{polynomial-time many-one reducible to $(C,D)$}, denoted by $(A,B) \leqmpp (C,D)$, if there is a function $h\in \FP$ such that $h(A)\subseteq C$ and $h(B)\subseteq D$.
The terms $\leqmpp$-completeness and -hardness also follow directly from this definition of reduction.

A disjoint pair $(A,B)$ is \emph{$\P$-separable}, if there exists a separator $S \in \P$ for $(A,B)$, i.e., a set $S$ such that $A \subseteq S$ and $B \subseteq \overline{S}$. A complexity class $\mathcal{C}$ has the \emph{shrinking property}, if for all $A, B \in \mathcal{C}$ there exist disjoint sets $A',B' \in \mathcal{C}$ such that $A' \subseteq A$, $B' \subseteq B$ and $A' \cup B' = A \cup B$. A complexity class $\mathcal{C}$ has the \emph{separation property}, if for all disjoint $A,B \in \mathcal{C}$ there exists an $S \in \mathcal{C} \cap \mathrm{co}\mathcal{C}$ that separates $A$ and $B$.

The model of deterministic Turing transducers can be extended to nondeterministic Turing transducers. Here, on input $x$, a nondeterministic Turing transducer can output one value on each  computation path.
Thus, a nondeterministic Turing transducer computes partial, multivalued functions. More precisely, a partial multivalued function $f$ is computed by a nondeterministic Turing transducer $N$ if and only if for every input $x$ the set of strings obtained by $f(x)$, i.e.\ $\{y \mid f(x) \mapsto y\}$, equals the set of strings written on the output tape of each accepting path of $N(x)$.
The class $\NPMV$ refers to the class of partial, multivalued functions $f$ that can be computed by a nondeterministic polynomial-time Turing transducer~\cite{bls84}. The class $\NPSV$ is the subset of $\NPMV$ that contains only single-valued partial functions~\cite{bls84}. The class $\NPkV$ is the subset of $\NPMV$ that contains only those partial multivalued functions that, for every input, output at most $k$ distinct values~\cite{ffnr96, nrrs98, ffnr03}. The class $\NPbV$ is the subset of $\NPMV$ that contains only those partial multivalued functions with values in the set $\{0,1\}$~\cite{ffnr96, ffnr03}. Given a function class $\mathcal{F}$, we denote the class of all total functions in $\mathcal{F}$ as $\mathcal{F}_t$. 
In order to compare function classes with each other we use $\subseteq _c$ as defined by Selman \cite{sel94}:
For partial, multivalued functions $f$ and $g$, say that $g$ is a \emph{refinement} of $f$, if, for all $x$, $g(x)$ is defined if and only if $f(x)$ is defined, and if $g(x)\mapsto y$, then $f(x)\mapsto y$. For function classes $\mathcal F$ and $\mathcal G$, we write $\mathcal F \subseteq_c \mathcal G$ if for every $f\in\mathcal F $ there exists a $g\in\mathcal G$ that is a refinement of $f$.

We follow the definition of reducibility on multivalued functions by Beyersdorff, Köbler and Messner~\cite{bkm09} and say that a multivalued function $h$ is polynomial-time many-one reducible to $g$, denoted by $h\leqmp g$, if there is a function $f\in \FP$ such that the set $\{y\mid h(x)\mapsto y\}$ of outputs computed by $h$ on input $x$ is equal to the set $\{y\mid g(f(x))\mapsto y\}$ of outputs computed by $g$ on input $f(x)$. The notions of $\leqmp$-completeness for class $\NPMVt$ follows directly from this definition of reducibility.

The class $\TFNP$ is the class of all total polynomial search problems (also known as total $\NP$ search problems)~\cite{mp91}. In terms of partial multivalued functions, $\TFNP$ can be defined as the class of functions $f\in \NPMVt$ such that the graph $\{(x,y)\mid f(x)\mapsto y\}$ is in $\P$. For the conjecture $\hTFNP$, Pudlák~\cite{pud17} refers to the following natural definition of reducibility between total polynomial search problems:
Let $R$ and $S$ be total polynomial search problems. We say that $R$ is \emph{polynomial-time many-one reducible} to $S$ if there exist $f,g\in \FP$ such that for all $x$ and $z$, it holds that $S(f(x))\mapsto z$ implies $R(x)\mapsto g(x, z)$.
The notion of many-one complete total polynomial search problems follows directly from the definition of reducibility.

Let $\SAT$ denote the set of satisfiable formulas and $\TAUT$ the set of tautologies. It is well known that $\SAT \in \NP$ and $\TAUT \in \coNP$.
We use the notion of proof systems for sets by Cook and Reckhow \cite{cr79}.
\begin{itemize}
\item A function $f \in \FP$ is called a \emph{proof system for $\img(f)$}.
\item We say that a proof system $g$ is \emph{(p-)simulated} by a proof system $f$, denoted by $f \leq g$ (resp., $f \leq ^{\mathrm{p}} g$), if there exists a total function $\pi$ (resp., $\pi \in \FP$) and a polynomial $p$ such that $|\pi(x)| \leq p(|x|)$ and $f(\pi(x))=g(x)$ for all $x$.
\item A proof system for $\TAUT$ is called a \emph{propositional proof system}.
\item We call a proof system $f$ \emph{(p-)optimal} for the set $\img(f)$, if $g \leq f$ (resp., $g \leq ^{\mathrm{p}} f$) for all $g \in \FP$ with $\img(g)=\img(f)$.
\end{itemize}

We can relativize each complexity and function class to some oracle $O$, by equipping all machines corresponding to the respective class with oracle access to $O$. That is, e.g., $\P^O \coloneqq \{ L(M^O) \mid \text{$M$ is a deterministic polynomial-time oracle Turing machine} \}$. The classes $\NP^O$, $\UP^O$ and so on are defined similarly. We can also relativize our notions of reducibility by using functions from $\FP^O$ instead of $\FP$. In other words, we allow the reduction functions to access the oracle in relativized instances. This results in polynomial-time many-one reducibilities relative to an oracle $O$, which we denote as $\leq_\mathrm{m}^{\mathrm{p},O}$ for sets and $\leq_\mathrm{m}^{\mathrm{pp},O}$ for pairs of disjoint sets.
In the same way, we can relativize (p-)simulation of proof systems to some oracle $O$, and denote the relativized simulation as $\leq^{O}$ resp. $\leq^{\mathrm{p},O}$.
When it is clear from context that some statements refer to the relativized ones relative to some fixed oracle $O$, we sometimes omit the indication of $O$ in the superscripts.

We define $p_i(n) \coloneqq n^i + i$.
Let $\{M_i\}_{i\in \mathbb N}$ and $\{F_i\}_{i\in\mathbb N}$ be, respectively, standard enumerations of nondeterministic polynomial-time (oracle) Turing machines resp.\ deterministic polynomial-time (oracle) Turing transducers, having the property that runtime of $M_i, F_i$ is bounded by $p_i$ relative to any oracle.
By standardness, $\{ L(M_i^O) \mid i\in\mathbb N\}=\NP^O$, $\{ F^O_i \mid i\in\mathbb N\}=\FP^O$.

We now take on the notations proposed by Dose and Glaßer~\cite{dg20} designed for the construction of oracles.
The domain of definition, image, and support for partial function $t\colon A\to \mathbb N$ are defined as $\dom(t) \coloneqq \{ x\in A \mid \exists y.t(x)=y \}$, $\img(t) \coloneqq \{ t(x) \mid x\in A \}$, $\supp(t) \coloneqq \{ x\in A \mid \exists y.t(x)=y>0 \}$.
We say that $t$ is \emph{injective on its support} if, for any $a,b\in\supp(t)$, $t(a)=t(b)$ implies $a=b$.
If $t$ is not defined at point $x$, then $t\cup \{ x\mapsto y\}$  denotes the extension $t'$ of $t$ that at $x$ has value $y$ and satisfies $\dom(t')=\dom(t) \cup \{x\}$.

For a set $A$, we denote with $A(x)$ the characteristic function at point $x$, i.e., $A(x)$ is 1 if $x\in A$, and $0$ otherwise.
We can identify an oracle $A\subseteq \mathbb N$ with its characteristic $\omega$-word $A(0)A(1)A(2)\cdots$ over $\Sigma^\omega$. In this way, $A(i)$ denotes both the characteristic function at point $i$ and the $i$-th character of its characteristic word.
A finite word $w$ describes an oracle which is partially defined, i.e., only defined for natural numbers $x<|w|$.
Occasionally, we understand $w$ as the set $\{i\mid w(i)=1\}$ and, e.g., we write $A=w\cup B$ where $A$ and $B$ are sets.
(However, we understand $|w|$ as the length of word $w$, and not the cardinality of set $\{i\mid w(i)=1\}$.)

In particular, for oracle machines $M$, the notation $M^w(x)$ refers to $M^{\{i\mid w(i)=1\}}(x)$ (that is, oracle queries that $w$ is not defined for are negatively answered).
This also allows us to define following notion: we say that $M^w(x)$ is \emph{definite} if all queries on all computation paths are $<|w|$ (or equivalently: $w(q)$ is defined for all queries $q$ on all computation paths); we say that $M^w(x)$ \emph{definitely accepts} (resp., \emph{definitely rejects}) if $M^w(x)$ is definite and accepts (resp., rejects).
This allows the following observation:
\begin{observation}\phantomsection\label{obs:partialoracles}
    \begin{enumerate}
        \item When $M^w(x)$ is a definite computation, and $v\sqsupseteq w$, then 
            $M^v(x)$ is definite. Computation $M^v(x)$ accepts if and only if $M^w(x)$ accepts.
        \item When $w$ is defined for all words of length $p_i(|x|)$, then $M_i^w(x)$ is definite.
        \item When $M^w(x)$ accepts on some computation path with set of oracle queries $Q$, and $w$, $v$ agree on $Q$, then $M^v(x)$ accepts on the same computation path and with the same set of oracle queries $Q$.
    \end{enumerate}
\end{observation}
For oracle $w$, transducer $F$, and machine $M$, we occasionally understand the notation $M^w(F^w(x))$ as the single computation of machine $M\circ F$ on input $x$ relative to $w$.
Consequently, we say that $M^w(F^w(x))$ definitely accepts (resp., rejects) when $M\circ F$ definitely accepts (resp., rejects) input $x$ relative to $w$.

In our oracle construction, we want to injectively reserve and assign countably infinitely many \emph{levels} $n$, that are, words of same length $n$, for a countably infinite family of witness languages, with increasingly large gaps.
For this, let $e(0) \coloneqq 2$, $e(i) \coloneqq 2^{e(i-1)}$.
There is a polynomial-time computable, polynomial-time invertible injective function $f$, mapping $(m,h)\in\mathbb N\times\mathbb N$ to $\mathbb N$.
Now define $H_m \coloneqq \{ e(f(m,h)) \mid h\in\mathbb N \}$ as the set of levels reserved for witness language $m$.
This definition ensures 
\begin{observation}\phantomsection\label{obs:leveldefinitions}
    \begin{enumerate}
        \item The set $H_m$ is countably infinite, a subset of the even numbers, and all $H_0, H_1, \dots$ are pairwise disjoint.
        \item The sequence $\min H_0, \min H_1, \dots$ is unbounded.
        \item When $n\in H_m$, then $n<n'<2^{n}$ implies $n'\not\in H_0, H_1, \dots$.
        \item Every set $H_m\in \P$ for all $m\in\mathbb N$.
\end{enumerate}
\end{observation}
Before the construction, we make the following combinatorial assertion:
\begin{lemma}\label{lemma:bipartite}
    Let $G$ be a directed bipartite graph with vertex parts $A$ and $B$.
    That is, every edge in $G$ is either from a vertex in $A$ to a vertex in $B$, or vice versa.
    Let $\Delta$ be an upper bound of the out-degree for every vertex in $G$.
    
    If $|A|,|B|>2\Delta$, then there exist $a\in A$ and $b\in B$ such that neither $(a,b)$ nor $(b,a)$ is an edge in $G$.
\end{lemma}
\begin{proof}
    Let $n=\min\{|A|,|B|\}>2\Delta$.
    Now remove vertices from $A$ and $B$ until both vertex parts each have precisely $n$ vertices, respectively, to form a directed bipartite graph $G'$ with vertex parts $A'$ and $B'$.
    Let $G''$ be the underlying undirected graph of $G'$.
    There are $\leq |A'|\cdot \Delta + |B'|\cdot\Delta<n^2$ many undirected edges in $G'$,
    but $n^2$ many undirected edges in the complete bipartite undirected graph $K_{n,n}$.

    This implies there exists $a\in A'\subseteq A$, $b\in B'\subseteq B$ that are not adjacent in $G''$; thus, both $(a,b)\not\in E(G')$ and $(b,a)\not\in E(G')$ for the induced directed bipartite subgraph $G'$.
    This also implies that for the original graph $G$, both $(a,b)\not\in E(G)$ and $(b,a)\not\in E(G)$.
\end{proof}

\section{Oracle Construction}\label{sec:construction}

We are primarily interested in an oracle $O$ with the property that relative to that oracle, $\hUP$, $\hDisjNP$, $\hDisjCoNP$, and $\neg\hNPcoNP$ hold, but our construction yields the following slightly stronger statements:
\begin{enumerate}
    \item $\NPcoNP=\P$ (implying $\neg\hNPcoNP$).
    \item $\DisjNP$ does not contain $\leqmpp$-hard pairs for $\DisjUP$ (implying $\hDisjNP$).
    \item $\UP$ does not contain $\leqmp$-complete languages (i.e., $\hUP$).
    \item $\DisjCoNP$ does not contain $\leqmpp$-hard pairs for $\DisjCoUP$ (implying $\hDisjCoNP$).
\end{enumerate}

Given a (possible partial) oracle $O$ and $m\in \mathbb N$, we define the following witness languages:
\begin{gather*}
    A_m^O \coloneqq \{ 0^n \mid n\in H_m, \text{there exists $x\in \Sigma^{n}$ such that $x\in O$ and $x$ ends with $0$} \}\\
    B_m^O \coloneqq \{ 0^n \mid n\in H_m, \text{there exists $x\in \Sigma^{n}$ such that $x\in O$ and $x$ ends with $1$} \}\\
    C_m^O \coloneqq \{ 0^n \mid n\in H_m, \text{there exists $x\in \Sigma^{n}$ such that $x\in O$} \}\\
    D_m^O \coloneqq \{ 0^n \mid n\in H_m, \text{for all $x\in \Sigma^{n}$, } x\in O \rightarrow \text{$x$ ends with $0$} \}\\
    E_m^O \coloneqq \{ 0^n \mid n\in H_m, \text{for all $x\in \Sigma^{n}$, } x\in O \rightarrow \text{$x$ ends with $1$} \}
\end{gather*}
This allows the following observation:
\begin{observation}\phantomsection\label{obs:witnesses}
    \begin{enumerate}
        \item If for all $n\in H_m$, $|O \cap \Sigma^{n}|\leq 1$, then $(A_m^O, B_m^O)$ is in $\DisjUP^O$, and $C_m^O$ is in $\UP^O$.
        \item If for all $n\in H_m$, $O \cap \Sigma^{n}$ contains at least one word but not two words with the same parity, (i.e., there exists $\alpha\in\Sigma^{n-1}0$, $\beta\in\Sigma^{n-1}1$ such that the set $O \cap \Sigma^{n}$ is equal to $\{\alpha\}$ or $\{\beta\}$ or $\{\alpha, \beta\}$), then $(D_m^O, E_m^O)$ is in $\DisjCoUP^O$.
    \end{enumerate}
\end{observation}

\paragraph{Preview of the construction.}
\begin{enumerate}[label=\arabic*.]
    \item Work towards $\P = \NPcoNP$: For all $a\neq b$, the construction tries to achieve that $M_a, M_b$ do not accept complementary.
        If this is not possible, $(M_a, M_b)$ inherently accept complementary, and thus $L(M_a)\in \NPcoNP$.
        Then, we start to \emph{encode} into the oracle, whether $M_a$ accepts some inputs or not.
        Thus, the final oracle will contain the encodings for almost all inputs, thus allowing to recover the accepting behavior of $M_a$ and hence to decide $L(M_i)$ in $\P$ using oracle queries.

    \item Work towards (ii), which implies $\hDisjNP$: For all $i\neq j$, the construction tries to achieve that $M_i, M_j$ do not accept disjointly, hence $(L(M_i), L(M_j))\not\in \DisjNP$.
        If this is not possible, $(M_i, M_j)$ inherently is a disjoint $\NP$-pair.
        In this case, we fix some $m$, make sure that $(A_m,B_m)$ is a disjoint $\UP$-pair and diagonalize against every transducer $F_r$, such that $F_r$ does not realize the reduction $(A_m, B_m)\leqmpp (L(M_i), L(M_j))$.
        This is achieved by, (i) for all $n\in H_m$, insert at most one word of length $n$ into $O$ (and thus $(A_m, B_m)\in\DisjUP$), and (ii) for every $r$ there is an $n\in H_m$ such that $0^{n}\in A_m$ but $M_i(F_r(0^n))$ rejects (or analogously $0^n\in B_m$ but $M_j(F_r(0^n))$ rejects).
    \item Work towards (iii), i.e., $\hUP$: Try to make $M_i$ accept on two separate paths.
        If this is not possible, then $L(M_i)$ inherently is a UP-language. In this case, we fix some $m$, make sure that $C_m$ is a language in $\UP$ and diagonalize against every transducer $F_r$ such that $F_r$ does not realize the reduction $C_m\leqmp L(M_i)$.
        This is achieved by, (i) for all $n\in H_m$, insert at most one word of length $n$ into $O$ (and thus $C_m\in\UP$), and (ii) for every $r$ there is an $n\in H_m$ such that $0^{n}\in C_m$ if and only if $M_i(F_r(0^n))$ rejects.
    \item Work towards (iv), which implies $\hDisjCoNP$: Try to achieve that that for some input, $M_i, M_j$ both reject.
        If this is not possible, $(M_i, M_j)$ inherently is a disjoint $\coNP$-pair.
        In this case, we fix some $m$, make sure that $(D_m,E_m)$ is a disjoint $\coUP$-pair and diagonalize against every transducer $F_r$, such that $F_r$ does not realize the reduction $(D_m, E_m)\leqmpp (L(M_i), L(M_j))$.
        This is achieved by, (i) for all $n\in H_m$, insert at least one word of length $n$ into $O$ but not two words with same parity (and thus $(D_m, E_m)\in\DisjCoUP$), and (ii) for every $r$ there is an $n\in H_m$ such that $0^{n}\in D_m$ but $M_i(F_r(0^n))$ accepts (or analogously $0^n\in E_m$ but $M_j(F_r(0^n))$ accepts).
\end{enumerate}
To these requirements, we assign following symbols representing tasks: $\tau^1_{a,b}$, $ \tau^2_{i,j}$, $ \tau^2_{i,j,r}$, $ \tau^3_{i}$, $ \tau^3_{i,r}$, $ \tau^4_{i,j}$, $ \tau^4_{i,j,r}$ for all $a,b,i,j,r\in \mathbb N, i\neq j$, $a\neq b$.
Symbol $\tau^1_{a,b}$ represents the coding or the destruction of $\NPcoNP$-pairs. Symbol $\tau^2_{i,j}$ represents the destruction of a disjoint $\NP$-pair, $\tau^2_{i,j,r}$ the diagonalization of that pair against transducer $F_r$.
Analogously for $\UP$ and $\tau^3_{i}, \tau^3_{i,r}$.
Analogously for $\DisjCoNP$ and $\tau^4_{i,j}, \tau^4_{i,j,r}$.

For the coding, we injectively define the code word $c(a,b,x)\coloneqq 0^a10^b10^l10^p1x$ with $p=p_a(|x|)+p_b(|x|)$, $l\in \mathbb N$ minimal such that $l\geq 7/8 |c(a,b,x)|$ and $c(a,b,x)$ has odd length. We call any word of the form $c(\cdot, \cdot, \cdot)$ a \emph{code word}. This ensures the following properties:
\begin{claim}\label{claim:codewords} For all $a,b\in\mathbb N$, $x\in\Sigma^*$
\begin{enumerate}
    \item $|c(a,b,x)| \not\in H_m$ for any $m$.
    \item For fixed $a,b$, the function $x \mapsto c(a,b,x)$ is polynomial-time computable, and polynomial-time invertible with respect to $|x|$.
    \item Relative to any oracle, the running times of $M_a(x)$ and $M_b(x)$ are both bounded by $<|c(a,b,x)|/8$.
    \item For every partial oracle $w\in\Sigma^*$, if $c(a,b,x)\leq |w|$, then $M_a^w(x)$ and $M_b^w(x)$ are definite.
\end{enumerate}
\end{claim}
\begin{proof}
    (i), (ii), and (iii) are immediate.
    For (iv), observe that for any oracle query $q$ on input $x$, it holds that
    \begin{equation*} |q|\leq \max\{p_a(|x|), p_b(|x|)\} \leq p_a(|x|)+ p_b(|x|) < |c(a,b,x)|, \end{equation*}
    hence, $q<c(a,b,x) \leq |w|$ and the computations are definite by definition.
\end{proof}

During the construction we successively add requirements that we maintain, which are specified by a partial function belonging to the set $\mathcal T$, and that are defined as follows: function $t\in\mathcal T$ if $t$ partially maps $\tau^1_{a,b}$, $\tau^2_{i,j}$, $\tau^3_i$, $\tau^4_{i,j}$ to $\mathbb N$, and $\dom(t)$ is finite, and $t$ is injective on its support.

A partial oracle $w\in\Sigma^*$ is called \emph{$t$-valid} for $t\in \mathcal T$ if it satisfies following requirements:
\begin{enumerate}
    \item[V1] If $t(\tau^1_{a,b}) = 0$, then there exists an $x$ such that $M_a^w(x)$, $M_b^w(x)$ both definitely accept or both definitely reject.\\ (Meaning: if $t(\tau^1_{a,b}) = 0$, then for every extension of the oracle, $M_a, M_b$ do not accept complementary.)
    \item[V2] If $0<t(\tau^1_{a,b})\leq c(a,b,x)< |w|$, then $M_a^w(x)$ is definite. Computation $M_a^w(x)$ accepts when $c(a,b,x)\in w$, and rejects when $c(a,b,x)\not\in w$.
        Note that when previous conditions are not met ($\tau^1_{a,b}\not\in \dom(t)$ or $t(\tau^1_{a,b})=0$ or $t(\tau^1_{a,b})>c(a,b,x)$) then the code word $c(a,b,x)$ may be a member of oracle $w$, independent of $M_a$, $M_b$.
        \\
        (Meaning: if $t(\tau^1_{a,b}) > 0$, then from $t(\tau^1_{a,b})$ on, we encode $L(M_a)$ into the oracle.
        That is, $L(M_a^O) = (\{ x\mid c(a,b,x)\in O \} \cup \text{some finite set}) \in \P^O$.)
    \item[V3] If $t(\tau^2_{i,j})=0$, then there exists $x$ such that $M_i^w(x)$, $M_j^w(x)$ both definitely accept.\\
        (Meaning: if $t(\tau^2_{i,j})=0$, then for every extension of the oracle, $(L(M_i), L(M_j))\not\in\DisjNP$.)
    \item[V4] If $t(\tau^2_{i,j})=m>0$, then for every $n\in H_m$ it holds that $|\Sigma^{n}\cap w|\leq 1$.\\
        (Meaning: if $t(\tau^2_{i,j})=m>0$, then ensure that $(A_m, B_m)\in\DisjUP$ relative to the final oracle.)
    \item[V5] If $t(\tau^3_{i})=0$, then there exists $x$ such that $M_i^w(x)$ is definite and accepts on two different paths.\\
        (Meaning: if $t(\tau^3_{i})=0$, then for every extension of the oracle, $L(M_i)\not\in\UP$.)
    \item[V6] If $t(\tau^3_{i})=m>0$, then for every $n\in H_m$ it holds that $|\Sigma^{n}\cap w|\leq 1$.\\
        (Meaning: if $t(\tau^3_{i})=m>0$, then ensure that $C_m\in\UP$ relative to the final oracle.)
    \item[V7] If $t(\tau^4_{i,j})=0$, then there exists $x$ such that $M_i^w(x)$, $M_j^w(x)$ both definitely reject.\\
        (Meaning: if $t(\tau^4_{i,j})=0$, then for every extension of the oracle, $(\overline{L(M_i)}, \overline{L(M_j)})\not\in\DisjCoNP$.)
\item[V8] If $t(\tau^4_{i,j})=m>0$, then for every $n\in H_m$ it holds that all words in $\Sigma^{n}\cap w$ have pairwise different parity. If additionally $w$ is defined for all words of length $n$, then $|\Sigma^{n}\cap w|>0$.\\
        (Meaning: if $t(\tau^4_{i,j})=m>0$, then ensure that $(D_m, E_m)\in\DisjCoUP$ relative to the final oracle.)
\end{enumerate}
Observe that V4, V6, V8 do not (pairwise) contradict each other, since $t$ is injective on its support and all $H_1, H_2, \dots$ are pairwise disjoint, by Observation~\ref{obs:leveldefinitions}(i).
Also observe that V2 and V4 (resp., V2 and V6, V2 and V8) do not contradict each other, as $c(\cdot, \cdot, \cdot)$ has odd length, but all $n$ in all $H_m$ are even by Observation~\ref{obs:leveldefinitions}(i).

The subsequent observation follows directly from the definition of $t$-valid partial oracles:
\begin{observation}\label{obs:downward}
    Let $t, t'\in\mathcal T$ such that $t'$ is an extension of $t$.
    Whenever $w\in\Sigma^*$ is $t'$-valid, then $w$ is $t$-valid. In particular, $w$ remains $t$-valid, even if $w$ contains code words $c(\cdot,\cdot,\cdot)$ that $t$ does not require to be in $w$, as long as V2 is satisfied.
\end{observation}

\paragraph{Oracle construction.} Let $T$ be a countable enumeration of 
\begin{alignat*}{3} \{\tau^1_{a,b} \mid a,b\in\mathbb N, a\neq b\} &\cup \{\tau^2_{i,j} \mid i,j\in\mathbb N, i\neq j\} &&\cup \{ \tau^2_{i,j,r} \mid i,j,r\in\mathbb N, i\neq j \} \\
&\cup \{\tau^3_{i} \mid i\in\mathbb N\} &&\cup \{ \tau^3_{i,r} \mid i,r\in\mathbb N \}  \\
&\cup \{\tau^4_{i,j} \mid i,j\in\mathbb N, i\neq j\} &&\cup \{ \tau^4_{i,j,r} \mid i,j,r\in\mathbb N, i\neq j \} 
\end{alignat*}
with the property that $\tau^2_{i,j}$ appears earlier than $\tau^2_{i,j,r}$, $\tau^3_{i}$ appears earlier than $\tau^3_{i,r}$, $\tau^4_{i,j}$ earlier than $\tau^4_{i,j,r}$.

We  recursively define an $\omega$-infinite sequence $\{(w_s, t_s)\}_{s<\omega}$, where the $s$-th term of the sequence is a pair $(w_s, t_s)$ of a partial oracle and a function in $\mathcal{T}$.
We call the $s$-th term the \emph{stage $s$}.

In each stage, we treat the smallest task in the order specified by $T$, and after treating a task we remove it and possibly other higher tasks from $T$.
In the next stage, we continue with the next task not already removed from $T$. (In every stage, there always exists a task not already removed, as we never remove \emph{all} remaining tasks from $T$ in any stage.)

We start with the nowhere defined function $t_0\in \mathcal T$ and the $t_0$-valid oracle $w_0\coloneqq\varepsilon$ as $0$-th stage.
Then we begin treating the tasks.

Thus, for stage $s>0$, we have that $w_0, w_1, \dots, w_{s-1}$ and $t_0, t_1, \dots, t_{s-1}$ are defined.
With this, we define the $s$-th stage $(w_s, t_s)$ such that (a) $w_{s-1} \sqsubsetneq w_s$, and $t_{s}\in \mathcal T$ is a (not necessarily strict) extension of $t_{s-1}$, and (b) $w_s$ is $t_s$-valid, and (c) the earliest task $\tau$ still in $T$ is treated and removed in some way.

So for each task we strictly extend the oracle and are allowed to add more requirements, by extending the valid function, that have to be maintained in the further construction.
Finally, we choose $O\coloneqq\bigcup_{i\in\mathbb N} w_i$.
(Note that $O$ is totally defined since in each step we strictly extend the oracle.)
Also, every task in $T$ is assigned some stage $s$ where it was treated (or removed from $T$).

We now define stage $s > 0$, which starts with 
some $t_{s-1}\in \mathcal T$ and a $t_{s-1}$-valid oracle $w_{s-1}$ and treats the first task that still is in $T$ choosing an extension $t_s\in\mathcal T$ of $t_{s-1}$ and a $t_s$-valid $w_s\sqsupsetneq w_{s-1}$.
Let us recall that each task is immediately deleted from $T$ after it is treated.
There are seven cases depending on the form of the task that is treated in stage $s$:
\bigskip

\textbf{Task} $\tau^1_{a,b}$: Let $t'\coloneqq t_{s-1}\cup \{\tau^1_{a,b} \mapsto 0\}$. If there exists a $t'$-valid $v\sqsupsetneq w_{s-1}$, then assign $t_s\coloneqq t'$ and let $w_s$ be the $<$-smallest (or equivalently, $\prec_\mathrm{lex}$-smallest) $t_s$-valid partial oracle $\sqsupsetneq w_{s-1}$.

Otherwise, let $t_s\coloneqq t_{s-1} \cup \{\tau^1_{a,b} \mapsto n \}$ with $n\in\mathbb N^+$ sufficiently large such that $n>|w_s|, \max\img(t_{s-1})$.
Thus $t_s$ is injective on its support, and $w_{s-1}$ is $t_s$-valid.
Let $w_s\coloneqq w_{s-1}y$ with $y\in\{0,1\}$ such that $w_s$ is $t_s$-valid.
We will show in Lemma~\ref{lemma:extension} that such $y$ does indeed exist.

(Meaning: try to ensure that $M_a, M_b$ do not accept complementary, cf. V1. If that is impossible, require that from now on the computations of $M_a$ are encoded into the oracle, cf. V2.)
\medskip

\textbf{Task} $\tau^2_{i,j}$: Let $t'\coloneqq t_{s-1}\cup \{\tau^2_{i,j} \mapsto 0\}$. If there exists $t'$-valid $v\sqsupsetneq w_{s-1}$, then assign $t_s\coloneqq t'$ and $w_s$ as the $<$-smallest $t_s$-valid partial oracle $\sqsupsetneq w_{s-1}$.
Remove all tasks $\tau^2_{i,j,0}, \tau^2_{i,j,1}, \dots$ from $T$.

Otherwise, let $t_s\coloneqq t_{s-1} \cup \{\tau^2_{i,j} \mapsto m\}$ with $m\in\mathbb N^+$ sufficiently large such that $m\not\in\img(t_{s-1})$ and that $w_{s-1}$ defines no word of length $\min H_m$. Thus $t_s$ is injective on its support, and $w_{s-1}$ is $t_s$-valid. Let $w_s\coloneqq w_{s-1}y$ with $y\in\{0,1\}$ such that $w_s$ is $t_s$-valid.
Again, we will show in Lemma~\ref{lemma:extension} that such $y$ does indeed exist.

(Meaning: try to ensure that $M_i, M_j$ do not accept disjointly, cf. V3. If that is impossible, choose a sufficiently large “fresh” $m$ and require for the further construction that $(A_m,B_m)\in\DisjUP$ (cf. V4). The treatment of tasks $\tau^2_{i,j,0}, \tau^2_{i,j,1}, \dots$ makes sure that $(A_m,B_m)$ cannot be reduced to $(L(M_i), L(M_j))$.)
\medskip

\textbf{Task $\tau^3_{i}$}: Let $t'\coloneqq t_{s-1}\cup \{\tau^3_{i} \mapsto 0\}$. If there exists $t'$-valid $v\sqsupsetneq w_{s-1}$, then assign $t_s\coloneqq t'$ and $w_s$ as the $<$-smallest $t_s$-valid partial oracle $\sqsupsetneq w_{s-1}$.
Remove all tasks $\tau^3_{i,0}, \tau^3_{i,1}, \dots$ from $T$.

Otherwise, let $t_s\coloneqq t_{s-1} \cup \{\tau^3_{i} \mapsto m\}$ with $m\in\mathbb N^+$ sufficiently large such that $m\not\in\img(t_{s-1})$ and that $w_{s-1}$ defines no word of length $\min H_m$. Thus $t_s$ is injective on its support, and $w_{s-1}$ is $t_s$-valid. Let $w_s\coloneqq w_{s-1}y$ with $y\in\{0,1\}$ such that $w_s$ is $t_s$-valid.
Again, we will show in Lemma~\ref{lemma:extension} that such $y$ does indeed exist.

(Meaning: try to ensure that $M_i$ does accept on two different paths, cf. V5. If that is impossible, choose a sufficiently large “fresh” $m$ and require for the further construction that $C_m\in\UP$ (cf. V6). The treatment of tasks $\tau^3_{i,0}, \tau^3_{i,1}, \dots$ makes sure that $C_m$ cannot be reduced to $L(M_i)$.)
\medskip

\textbf{Task $\tau^4_{i,j}$}: Defined symmetric. Let $t'\coloneqq t_{s-1}\cup \{\tau^4_{i,j} \mapsto 0\}$. If there exists $t'$-valid $v\sqsupsetneq w_{s-1}$, then assign $t_s\coloneqq t'$ and $w_s$ as the $<$-smallest $t_s$-valid partial oracle $\sqsupsetneq w_{s-1}$.
Remove all tasks $\tau^4_{i,j,0}, \tau^4_{i,j,1}, \dots$ from $T$.

Otherwise, let $t_s\coloneqq t_{s-1} \cup \{\tau^4_{i,j} \mapsto m\}$ with $m\in\mathbb N^+$ sufficiently large such that $m\not\in\img(t_{s-1})$ and that $w_{s-1}$ defines no word of length $\min H_m$. Thus $t_s$ is injective on its support, and $w_{s-1}$ is $t_s$-valid. Let $w_s\coloneqq w_{s-1}y$ with $y\in\{0,1\}$ such that $w_s$ is $t_s$-valid.
Again, we will show in Lemma~\ref{lemma:extension} that such $y$ does indeed exist.

(Meaning: try to ensure that $M_i, M_j$ do not reject disjointly, cf. V7. If that is impossible, choose a sufficiently large “fresh” $m$ and require for the further construction that $(D_m,E_m)\in\DisjCoUP$ (cf. V8). The treatment of tasks $\tau^4_{i,j,0}, \tau^4_{i,j,1}, \dots$ makes sure that $(D_m,E_m)$ cannot be reduced to $(\overline{L(M_i)}, \overline{L(M_j)})$.)
\medskip

\textbf{Task} $\tau^2_{i,j,r}$:  We have $t_{s-1}(\tau^2_{i,j}) = m\in\mathbb N^+$. Let $t_s \coloneqq  t_{s-1}$ and choose $t_s$-valid $w_s\sqsupsetneq w_{s-1}$
such that there is some $n\in\mathbb N$ and one of the following holds:
\begin{itemize}[label=–,nosep]
    \item $0^n\in A_m^v$ for all $v\sqsupseteq w_s$ and $M_i^{w_s}(F_r^{w_s}(0^n))$ definitely rejects.
    \item $0^n\in B_m^v$ for all $v\sqsupseteq w_s$ and $M_j^{w_s}(F_r^{w_s}(0^n))$ definitely rejects.
\end{itemize}
In Theorem~\ref{thm:disjnp} we show that such $w_s$ does exist.

(Meaning: ensure that $F_r$ does not reduce $(A_m, B_m)$ to $(L(M_i), L(M_j))$.)
\medskip

\textbf{Task} $\tau^3_{i,r}$:  We have $t_{s-1}(\tau^3_{i}) = m\in\mathbb N^+$. Let $t_s \coloneqq  t_{s-1}$ and choose $t_s$-valid $w_s\sqsupsetneq w_{s-1}$
such that there is some $n\in\mathbb N$ and one of the following holds:
\begin{itemize}[label=–,nosep]
    \item $0^n\in C_m^v$ for all $v\sqsupseteq w_s$ and $M_i^{w_s}(F_r^{w_s}(0^n))$ definitely rejects.
    \item $0^n\not\in C_m^v$ for all $v\sqsupseteq w_s$ and $M_i^{w_s}(F_r^{w_s}(0^n))$ definitely accepts.
\end{itemize}
In Theorem~\ref{thm:up} we show that such $w_s$ does exist.

(Meaning: ensure that $F_r$ does not reduce $C_m$ to $L(M_i)$.)
\medskip

\textbf{Task $\tau^4_{i,j,r}$}: Defined symmetric to $\tau^2_{i,j,r}$. Choose $t_s$-valid $w_s\sqsupsetneq w_{s-1}$ such that for some $n\in\mathbb N$, one of the two holds:
\begin{itemize}[label=–,nosep]
    \item $0^n\in D_m^v$ for all $v\sqsupseteq w_s$ and $M_i^{w_s}(F_r^{w_s}(0^n))$ definitely accepts.
    \item $0^n\in E_m^v$ for all $v\sqsupseteq w_s$ and $M_j^{w_s}(F_r^{w_s}(0^n))$ definitely accepts.
\end{itemize}
In Theorem~\ref{thm:disjconp} we show that such $w_s$ does exist.

(Meaning: ensure that $F_r$ does not reduce $(D_m, E_m)$ to $(\overline{L(M_i)},\overline{L(M_j)})$.)

\bigskip

Observe that $t_s$ is always defined to be in $\mathcal T$.
Remember that the treated task is immediately deleted from $T$.
This completes the definition of stage $s$, and thus, the entire sequence $\{(w_s, t_s)\}_{s<\omega}$.

We now show that this construction is indeed possible, by stating and proving the theorems/lemma that were announced in the definition.
First we state a simple observation from the construction, concerning tasks $\tau^1_{a,b}$ resp. pairs of machines that accept complementary:
\begin{lemma}\label{lemma:npconp-tasks}
    Let $s\in \mathbb N^+$, $(w_0, t_0), \dots, (w_s, t_s)$ defined, $w\in \Sigma^*$ be a $t_s$-valid oracle with $w\sqsupseteq w_s$, and $z\coloneqq c(a,b,x)$ for some $a,b,x$.

    Suppose that $0<t_s(\tau^1_{a,b})\leq z\leq |w|$.
    Then $M_a^w(x)$, $M_b^w(x)$ are definite, and $M_a^w(x)$ accepts if and only if $M_b^w(x)$ rejects.
\end{lemma}
\begin{proof}
    First, we observe that $M_a^w(x)$, $M_b^w(x)$ are definite: this immediately follows from Claim~\ref{claim:codewords}(iv).

    To argue for the claimed equivalence, assume it does not hold.
    Assume that  $M_a^w(x)$ and $M_b^w(x)$ accept (resp., reject).
    Let $s'\leq s$ be the stage that treated $\tau^1_{a,b}$.
    Such stage exists, as otherwise $t_s(\tau^1_{a,b})$ is undefined.
    Let $t'\coloneqq t_{s'-1}\cup \{\tau^1_{a,b}\mapsto 0\}$.

    We show that $w$ is $t'$-valid and $w\sqsupsetneq w_{s'-1}$.
    Last assertion is immediate, since $w_{s'-1}\sqsubsetneq w_{s'} \sqsubseteq w_s \sqsubseteq w$.
    Also note that $w$ is $t_{s'-1}$-valid, since $w$ is $t_s$-valid by hypothesis of this Lemma and thus Observation~\ref{obs:downward} applies.

    Hence, for $t'$-validity, only V1 is at risk, but by assumption, both $M_a^w(x)$ and $M_b^w(x)$ definitely accept (resp., reject).

    Thus, $w$ is a possible $t'$-valid extension of $w_{s'-1}$ in stage $s'$. We obtain that the treatment of task $\tau^1_{a,b}$ would define $t_{s'}=t'$.
    But then $t_s(\tau^1_{a,b})=0$, contradicting the hypothesis of the Lemma.
\end{proof}

Now, we describe how a valid oracle can by extended by one bit such that it remains valid:
\begin{lemma}\label{lemma:extension}
    Let $s\in \mathbb N$, $(w_0, t_0), \dots, (w_s, t_s)$ defined, and let $w\in \Sigma^*$ be a $t_s$-valid oracle with $w\sqsupseteq w_s$, and $z\coloneqq |w|$. (Think of $z$ as the next word we need to decide its membership to the oracle, i.e., $z\not\in w0$ or $z\in w1$.)
    Then there exists $y\in\{0,1\}$ such that $wy$ is $t_s$-valid. Specifically:
    \begin{enumerate}

        \item If $z=c(a,b,x)$ and $0<t_s(\tau^1_{a,b})\leq z$, then $w1$ is $t_s$-valid if $M_a^w(x)$ accepts (or when $M_b^w(x)$ rejects), and $w0$ is $t_s$-valid if $M_a^w(x)$ rejects (or when $M_b^w(x)$ accepts).\\
            (Meaning: if we are at a position of some mandatory code word, add the word as appropriate for the $\NPcoNP$-pair.)
        \item If there exists $\tau=\tau^2_{i,j}$ or $\tau=\tau^3_i$ with $m=t_s(\tau)>0$ and $n\in H_m$ such that $|z|=n$, $w\cap \Sigma^{n} \neq \emptyset$, then $w0$ is $t_s$-valid.\\
            (Meaning: if we are on a level $n$ belonging to a $\DisjUP$-pair or a $\UP$-language, ensure that there is no more than one word on that level.)
        \item If there exists $\tau^4_{i,j}$, $m=t_s(\tau^4_{i,j})>0$ and $n\in H_m$ such that $|z|=n$ and there is some other word $x\in w\cap \Sigma^n$ with same parity as $z$, then $w0$ is $t_s$-valid.
            (Meaning: if we are on a level $n$ belonging to a $\DisjCoUP$-pair, ensure that on that level, there are no two words with the same parity.)
        \item If there exists $\tau^4_{i,j}$, $m=t_s(\tau^4_{i,j})>0$ and $n\in H_m$ such that $|z|=n$, $|z+1|>n$, $w\cap \Sigma^{n} = \emptyset$, then $w1$ is $t_s$-valid.\\
            (Meaning: if we finalize level $n$ belonging to a $\DisjCoUP$ witness pair, ensure that there is at least one word on that level.)
        \item In all other cases, $w0$ and $w1$ are $t_s$-valid.
    \end{enumerate}
\end{lemma}
\begin{proof}
    Suppose that $w$ is $t_s$-valid but $wy$ is not $t_s$-valid. Then V1, V3, V5 and V7 are not responsible for this, because $w$ is $t_s$-valid, $wy \sqsupseteq w$ and V1, V3, V5 and V7 are statements about definite computations. If they are true with respect to $w$, they also have to be true with respect to $wy$, since otherwise the computations would not have been definite.
    This means that one of the requirements V2, V4, V6 or V8 is violated with respect to $t_s$-validity of $wy$.
    We will see that the violation of any requirement V2, V4, V6, V8 leads to a contradiction, hence $wy$ is $t_s$-valid.

    Assume V2 is violated. Then $0<t_s(\tau^1_{a,b})\leq c(a,b,x)<|wy|$ for suitable $a,b,x$ and $M_a^{wy}(x)$ is definite, and accepts if and only if $c(a,b,x)\not\in wy$. If $c(a,b,x)\neq z$, then $c(a,b,x)<|w|$ and V2 is violated with respect to $t_s$-validity of $w$. This contradicts assumption that $w$ is $t_s$-valid.
    Hence $z=c(a,b,x)$. 
    The lemma treats this case by part (i). We have that $y$ is chosen such that $c(a,b,x)\in wy$ if and only if $M_a^w(x)$ accepts.
    Note that $c(a,b,x)=|w|$ and conditions of Claim~\ref{claim:codewords}(iv) apply: $M_a^w(x)$ is definite. Hence $M_a^w(x)$ accepts if and only if $M_a^{wy}(x)$ accepts, by Observation~\ref{obs:partialoracles}(i). Thus we obtain the contradiction
    \[ c(a,b,x)\in wy \iff M_a^w(x) \text{ accepts} \iff M_a^{wy}(x) \text{ accepts} \iff c(a,b,x)\not\in wy, \]
    where the last equivalence holds by assumption.

    Assume V4 is violated. Then $t_s(\tau^2_{i,j}) = m > 0$ for some $\tau^2 _{i,j}$ and for some $n \in H_m$, we have $|\Sigma ^n \cap wy | > 1$. If $|z|\neq n$, then $|\Sigma ^n \cap w| = |\Sigma ^n \cap wy | > 1$, and V2 is violated with respect to $t_s$-validity of $w$. This contradicts assumption that $w$ is $t_s$-valid.
    Hence $|z|=n$.
    The Lemma treats $z$ by part (ii), setting $y=0$. We have $|\Sigma ^n \cap w|=|\Sigma ^n \cap w0| > 1$. Again, this contradicts the hypothesis that $w$ is $t_s$-valid.
    Symmetric if V6 is violated.

    Assume V8 is violated. Then $t_s(\tau^4_{i,j}) = m > 0$ for some $\tau^4 _{i,j}$ and for some $n \in H_m$, 
    one of the two holds: (a) there are two words $x, x'\in \Sigma^n\cap wy$ having the same parity, or (b) $wy$ is defined for all words of length $n$ but $|\Sigma^n\cap wy|=0$.
    Again we can suppose that $|z|=n$, as otherwise $w$ is not $t_s$-valid.

    If (a) holds, then the Lemma treats $z$ by part (iii). Since $y=0$, we have that $x, x'\in \Sigma^n\cap w$ having same parity, thus $w$ violates V8. This contradicts the hypothesis that $w$ is $t_s$-valid. 

    If (b) holds, then $|z+1|>n$ and the Lemma treats $z$ by part (iv). Since $y=1$, we obtain with above assumption the contradiction $0=|\Sigma^n\cap wy|=|\Sigma^n\cap w1|\geq 1$.
\end{proof}

Now we show that the construction is possible for  $\tau^2_{i,j,r}$, $\tau^3_{i,r}$ and $\tau^4_{i,j,r}$, respectively.
We first consider task $\tau^2_{i,j,r}$.

\begin{theorem}\label{thm:disjnp}
    Let $s\in \mathbb N^+$, $(w_0, t_0), \dots, (w_{s-1}, t_{s-1})$ defined.
    Consider task $\tau^2_{i,j,r}$. 

    Suppose that $t_s=t_{s-1}$, $t_s(\tau^2_{i,j})=m>0$.
    Then there exists a $t_{s}$-valid $w\sqsupsetneq w_{s-1}$ and $n\in\mathbb N$ such that one of the two holds:
    \begin{enumerate}[noitemsep]
        \item $0^n\in A_m^v$ for all $v\sqsupseteq w$ and $M_i^{w}(F_r^{w}(0^n))$ definitely rejects.
        \item $0^n\in B_m^v$ for all $v\sqsupseteq w$ and $M_j^{w}(F_r^{w}(0^n))$ definitely rejects.
    \end{enumerate}
\end{theorem}
\begin{proof}
Let us fix  $i,j,r$ throughout the proof of the theorem.

Let $\hat{s}<s$ be the stage that treated $\tau^2_{i,j}$.
Such stage exists, as otherwise $t_{s}(\tau^2_{i,j})$ is undefined.
We have $m=t_{\hat{s}}(\tau^2_{i,j})=t_s(\tau^2_{i,j})$; fix $m$ for the rest of the proof.

We assume that for all $t_{s}$-valid $w\sqsupsetneq w_{s-1}$, neither (i) nor (ii) holds.
From this we will deduce a contradiction, by constructing a suitable oracle $u'\sqsupsetneq w_{\hat{s}-1}$, which is valid with respect to $t' \coloneqq  t_{\hat{s}-1}\cup \{\tau^2_{i,j}\mapsto 0\}$.
Then, by definition, we obtain that $u'$ is a possible $t'$-valid extension of $w_{\hat{s}-1}$ in stage $\hat{s}$, hence $t_{\hat{s}}=t'$, contradicting the hypothesis of this Theorem~\ref{thm:disjnp}.

Let
\begin{equation*} \gamma(n) \coloneqq  \max(p_i(p_r(n))+p_r(n), p_j(p_r(n))+p_r(n)) \end{equation*}
be the polynomial bounding the runtime of $M_i\circ F_r$, $M_j\circ F_r$ with respect to input length $n$ relative to any oracle.
This implies that whenever some partial oracle $u'$ is defined for all words of length $\leq \gamma(n)$, then $M_i^{u'}(F_r^{u'}(x))$, $M_j^{u'}(F_r^{u'}(x))$ are definite for all inputs $x\in\Sigma^n$.
Let us define $n\in\mathbb N^+$ as the smallest $n\in H_m$ such that $w_{s-1}$ does not define any words of length $\geq n$,
and
\begin{equation}\label{eq:expbound}
    2^n > \gamma(n),\quad  2^{n-1} > 8\gamma(n)\tag{$\ast$}
\end{equation}
The first inequality of \eqref{eq:expbound} ensures that no level $n<n'\leq \gamma(n)$ is reserved for any witness language, that is $n'\not\in H_0, H_1, \dots$ (cf. Observation~\ref{obs:leveldefinitions}(iii)). The second inequality ensures that there are enough words of length $n$ such that certain combinatorial arguments work.

For the remaining proof, we additionally fix $n$.
Observe that $\ell(Q)\leq\gamma(n)$ for $Q$ being the set of oracle queries asked by either the computation $M_i(F_r(0^n))$ or the computation $M_j(F_r(0^n))$.
We define $u\sqsupseteq w_{s-1}$ as the $<$-minimal $t_{s}$-valid partial oracle that is defined precisely for all words up to length $<n$. Such oracle exists by Lemma~\ref{lemma:extension}, by extending $w_{s-1}$ bitwise such that it remains $t_{s}$-valid.

For our proof, we are not considering all $t_s$-valid $w$, but rather a sufficient subset of those. 
For $X \subseteq \Sigma^n, |X|\leq 1$, we will define a $t_s$-valid partial oracle $u(X)\sqsupsetneq u$ that is defined for all words of length $\leq\gamma(n)$, and such that $u(X) \cap \Sigma^n = X$, i.e., $u(X)$ and $X$ agree on $\Sigma^n$.
For each $u(X)$, we define a directed graph $G(X)$ on vertex set $\Sigma^{\leq\gamma(n)}$.
Graph $G(X)$ captures the dependencies (caused by the oracle queries) that computations represented by code words of length $> n$ have on other words in the oracle. We use a graph to model these dependencies, because direct dependencies are captured by edges and transitive dependencies are captured by paths in the graph. The set of words the membership $z\in u(X)$ of a single word $z\in\Sigma^{\leq\gamma(n)}$ (transitively) depends on, is exactly $R_{G(X)}(z)$, i.e., all words that are reachable from $z$ in $G(X)$. 

\paragraph{Definition of $u(X)$, $G(X)$:} Let $X\subseteq\Sigma^n$ with $|X|\leq 1$. We construct $u(X)$ and $G(X)=(V,E)$ inductively. Fix the vertex set $V=\Sigma^{\leq\gamma(n)}$.
Basis clauses:
\begin{enumerate}[label=(\arabic*)]
    \item For $z\in \Sigma^{<n}$, let $z\in u(X)$ if and only if $z\in u$.
    \item For $z\in \Sigma^{n}$, let $z\in u(X)$ if and only if $z\in X$.
\end{enumerate}
Inductive clauses: Let $z\in\Sigma^{\leq \gamma(n)}$, $|z|>n$, and $u(X)$ defined for words $<z$.
\begin{enumerate}[label=(\arabic*),resume*]
    \item If $z=c(a,b,x)$ for suitable $a,b,x$ with $0<t_s(\tau^1_{a,b})\leq z$, continue as follows:

        Mark vertex $z$ as \emph{active code word} in $G(X)$.
        Define $z\in u(X)$ if $M_a^{{u(X)}}(x)$ accepts.
        Then, let $(z,q)\in E$ for all oracle queries $q$ on the leftmost accepting paths of $M_a(x)$ and $M_b(x)$ relative to $u(X)$, if any exists. (We will see that precisely one of these machines accepts.)

        (Meaning: If $z$ is a mandatory code word for the coding of $\tau_{a,b}^1$ in combination with V2, we can construct $u(X)$ for $z$ like in Lemma~\ref{lemma:extension}. Later we need to know which words the leftmost accepting paths of $M_a(x)$ resp. $M_b(x)$ relative to $u(X)$ depend on. These dependencies are captured by adding respective edges to the edge set of $G(X)$.)

    \item Otherwise, $z\not\in u(X)$.
\end{enumerate}
Extremal clause: (5) No other edges are in $E$.
\medskip

Note that the conditions of $z$ being marked as active code word in $G(X)$ are independent of $X$, hence we can just say that $z$ is (not) an active code word without specifying the corresponding graph, meaning that $z$ is (not) an active code word in all $G(\cdot)$.

We now make some claims concerning $u(X)$ and $G(X)$.

\begin{claim}\label{claim:welldefined}
    Let $X\subseteq\Sigma^n$, $|X|\leq 1$.
    \begin{enumerate}
        \item $u(X)$ is well-defined, is defined for all words of length $\leq \gamma(n)$, $u(X) \cap \Sigma^n = X$, and $u(X) \sqsupsetneq u\sqsupseteq w_{s-1}$.
        \item $u(X)$ is $t_s$-valid.
        \item If $z=c(a,b,x)$ is an active code word,
            then the following statements are equivalent: (a) $z\in u(X)$, (b) $M_a^{u(X)}(x)$ accepts, (c) $M_b^{u(X)}(x)$ rejects.
        \item $G(X)$ forms a directed acyclic graph (which is not necessarily connected).
            In particular, for every directed edge from vertex $a$ to $b$, it holds that $a>b$.
    \end{enumerate}
\end{claim}
\begin{proof}
    To (i): First note that the definition of $u(X)$ is well-defined, particularly the set of oracle queries of $M_a(x)$ resp. $M_b(x)$ relative to $u(X)$, like in clause (3).
    When $u=u(X)$ is defined up to word $<z=c(a,b,x)$, then $M_a^{{u(X)}}(x)$, $M_b^{{u(X)}}(x)$ are definite by Claim~\ref{claim:codewords}(iv), hence cannot ask queries $\geq z$.

    The remaining assertions immediately follow from definition.
    \medskip

    To (ii): Let $u_0, u_1, u_2, \dots, u_l$ be a length-ordered enumeration of all prefixes of $u(X)$ that are defined for at least all words of length $<n$, that is, $u = u_0 \sqsubsetneq u_1 \sqsubsetneq u_2 \sqsubsetneq \dots \sqsubsetneq u_l = u(X)$.
    We show inductively that every $u_k$, $0\leq k\leq l$, is $t_s$-valid.
    Thus, we obtain that $u_l=u(X)$ is $t_s$-valid. 
    Base case is immediate, as $u_0 = u$ is $t_s$-valid by choice.

    For the inductive case from $u_k$ to $u_{k+1}$, let $z=|u_k|$, and $y=(u(X))(z)$. It holds that $u_ky=u_{k+1}\sqsubseteq u(X)$, i.e., $y$ is the last bit of $u_{k+1}$.
    By induction hypothesis, $u_k$ is $t_s$-valid.
    Note that $y$ is defined by above inductive definition by one of the clauses (2)--(4).

    We need to show that $u_ky=u_{k+1}$ is $t_s$-valid. For this, we will employ Lemma~\ref{lemma:extension} with respect to $t_s$-valid $u_k$.
    We analyze three cases on which clause (2), (3) or (4) defines $y$.
    \medskip

    Clause (3): We have $z=c(a,b,x)$ for suitable $a,b,x$, and case~\ref{lemma:extension}(i) applies.
    Since we have $y=1$ if and only if $M_a^w(x)$ accepts, we obtain that $u_ky$ is $t_s$-valid.\medskip

    Clause (4): Note that $n<|z|<2^n$ by \eqref{eq:expbound}, hence by Observation~\ref{obs:leveldefinitions}(iii) the cases~\ref{lemma:extension}(ii--iv) cannot apply.
    Furthermore, case~\ref{lemma:extension}(i) cannot apply either. Otherwise we have $z= c(a,b,x)$ for suitable $a,b,x$ and $0<t_{s}(\tau^1_{a,b})\leq z$. However, this case would be handled in clause (3).
    In total, only case~\ref{lemma:extension}(v) applies, hence $u_ky$ is $t_s$-valid.\medskip

    Clause (2):
    The case~\ref{lemma:extension}(i) cannot apply, since otherwise we have $z= c(\cdot,\cdot,\cdot)$ but by clause (2),  $|z|=n\in H_m$, contradicting Claim~\ref{claim:codewords}(i).
    The cases~\ref{lemma:extension}(iii--iv) cannot apply, since otherwise $t_s(\tau)=m'>0$ for some $\tau\neq \tau^2_{i,j}$ and $|z|=n\in H_{m'}$. Since $t_s$ is injective on its support, $m'=t_s(\tau)\neq t_s(\tau^2_{i,j})=m$. But then $n\in H_m, H_{m'}$, which contradicts Observation~\ref{obs:leveldefinitions}(i) that $H_m, H_{m'}$ are disjoint.
    Only two cases~\ref{lemma:extension}(ii) and (v) remain. If case~\ref{lemma:extension}(v) applies, then $u_{k+1}$ is $t_s$-valid and we are done.

    Therefore, assume that~\ref{lemma:extension}(ii) applies.
    From the conditions of the case follows that $u_k\cap\Sigma^n\neq\emptyset$.
    This means that $y=0$ as otherwise $|u_k1\cap\Sigma^n|=|u_k\cap\Sigma^n|+1>1$ and by previous Claim~\ref{claim:welldefined}(i) it holds that $|X|>1$, contradicting the condition of this Claim~\ref{claim:welldefined-disjconp} that $|X|\leq 1$.
    As $y=0$, case~\ref{lemma:extension}(ii) asserts that $u_k0=u_{k+1}$ is $t_s$-valid.
    \medskip

    To (iii):
    Equivalence of (a) and (b) immediately follows from definition. We prove that (b) and (c) are equivalent.
    From the previous Claim~\ref{claim:welldefined}(ii), it follows that $u(X)$ is $t_s$-valid, hence $t_{s-1}$-valid.
    Also, previous Claim~\ref{claim:welldefined}(i) shows $u(X)\sqsupseteq w_{s-1}$.
    Conditions of Lemma~\ref{lemma:npconp-tasks} (invoked with regard to $t_{s-1}$-valid $u(X)\sqsupseteq w_{s-1}$) apply: We have $0<t_{s-1}(\tau^1_{a,b})=t_{s}(\tau^1_{a,b})\leq z< |u(X)|$ by definition of clause (3), and $u(X)\sqsupseteq w_{s-1}$.
    Hence $M_a^{u(X)}(x)$ accepts if and only if $M_b^{u(X)}(x)$ rejects.
    \medskip

    To (iv): Edges are only added in clause (3).
    Here Claim~\ref{claim:codewords}(iii) asserts that $q<z$ for edges $(z,q)$.
\end{proof}

Since we assume that neither (i) nor (ii) of the Theorem~\ref{thm:disjnp} holds for any $t_{s}$-valid $w\sqsupsetneq w_{s-1}$, we have:
\begin{itemize}[noitemsep]
    \item Whenever $\alpha\in\Sigma^n$ ends with $0$, then $M_i(F_r(0^n))$ accepts relative to $u(\{\alpha\})$, 
    \item Whenever $\beta\in\Sigma^n$ ends with $1$, then $M_j(F_r(0^n))$ accepts relative to $u(\{\beta\})$.
\end{itemize}

For any $\alpha\in\Sigma^{n-1}0$, define $Q_\alpha$ as the set of oracle queries of the leftmost accepting path of $M_i(F_r(0^n))$ relative to $u(\{\alpha\})$.
Define the set $Q_\beta$ analogously as the set of oracle queries of the leftmost accepting path of $M_j(F_r(0^n))$ relative to $u(\{\beta\})$.
Let $Q^+_\alpha\coloneqq R_{G({\{\alpha\}})}(Q_\alpha)$ be $Q_\alpha$'s closure under successor in $G({\{\alpha\}})$. Define $Q^+_\beta$ respectively.

As was outlined in the beginning, we will find a partial oracle $u'$ such that both $M_i(F_r(0^n))$ and $M_j(F_r(0^n))$ accept relative to $u'$.
Remember that $t'=t_{\hat{s}-1}\cup \{\tau^2_{i,j}\mapsto 0\}$.
If we show that $u'\sqsupsetneq w_{\hat{s}-1}$ is $t'$-valid, then $\tau^2_{i,j,r}$ was removed in stage $\hat{s}$, contradicting the assumption that it was treated in stage $s>\hat{s}$.

In order to freeze the accepting paths, we want to freeze their respective oracle queries.
For this, we need to find $\alpha\in\Sigma^{n-1}0$, $\beta\in\Sigma^{n-1}1$ such that $u(\{\alpha\})$ and $u(\{\beta\})$ do not ``spoil'' each other, such that the two oracles can be ``merged'' into desired $u'$.
We say that $\alpha\in\Sigma^{n-1}0$ \emph{spoils} $\beta\in\Sigma^{n-1}1$ and vice versa when $u(\{\alpha\})$ and $u(\{\beta\})$ do not agree on $Q^+_\alpha\cap Q^+_\beta$. (\emph{Spoilage} is a symmetric relation.)

There exists a pair $\alpha\in\Sigma^{n-1}0$, $\beta\in\Sigma^{n-1}1$ such that $\alpha$ and $\beta$ do not spoil each other; we postpone the proof to the end.
For now, fix these words $\alpha$ and $\beta$.

For brevity, we write $\mathcal{V}$ as the set of $t_{\hat{s}-1}$-valid partial oracles.
We now show that there is a suitable oracle $u'\in \mathcal{V}$ such that both $M_i(F_r(x))$ and $M_j(F_r(x))$ definitely accept relative to $u'$, by “merging” $u(\{\alpha\})$ and $u(\{\beta\})$ into $u'$.

\begin{claim}\label{claim:finalextension-disjnp}
    There exists $u'\sqsupsetneq u$ with $u'\in\mathcal V$ (that is, $t_{\hat{s}-1}$-valid), and is defined for all words of length $\leq \gamma(n)$, and such that
    \begin{enumerate}[noitemsep]
        \item $u'$ agrees with $u(\{\alpha\})$ on $Q^+_\alpha$, and 
        \item $u'$ agrees with $u(\{\beta\})$ on $Q^+_\beta$.
    \end{enumerate}
    With Observation~\ref{obs:partialoracles}(ii) and (iii), this means that both $M_i(F_r(0^n))$ and $M_j(F_r(0^n))$ definitely accept relative to $u'$.
\end{claim}
\begin{proof}
    We iteratively extend $u_0\coloneqq u$ bitwise using Lemma~\ref{lemma:extension} to $u_1, u_2, \dots$ such that it remains in $\mathcal{V}$ and at relevant positions agrees with $u(\{\alpha\})$ resp. $u(\{\beta\})$ such that (i) and (ii) are satisfied, until $u_k$ is sufficiently long. We proceed inductively, and maintain the induction statement that 
    \begin{itemize}[nosep]
        \item for every $i>0$, $u_k=u_{k-1}y$ for some $y\in\{0,1\}$, i.e., every $u_k$ extends the previous $u_{k+1}$ by a single bit,
        \item $u_k\in\mathcal V$,
        \item for all words $q\in Q^+_\alpha$ defined by $u_k$, it holds that $q\in u_k \iff q\in u(\{\alpha\})$,
        \item for all words $q\in Q^+_\beta$ defined by $u_k$, it holds that $q\in u_k \iff q\in u(\{\beta\})$,
    \end{itemize}
    For some sufficiently large $l$, $u'\coloneqq u_l$ is defined for all words of length $\leq\gamma(n)$. Also, every word $q\in Q^+_\alpha\cup Q^+_\beta$ has length $\leq \gamma(n)$, hence is defined by $u_l$.
    Invoking the induction statement, $u'$ agrees with $u(\{\alpha\})$ on $Q^+_\alpha$, and $u'$ agrees with $u(\{\beta\})$ on $Q^+_\beta$, as desired.

    Base case with $u_0=u$ is immediate since with Observation~\ref{obs:downward}, we have $u\in\mathcal{V}$. Also by construction $u$ agrees with $u(\{\alpha\})$ and agrees with $u(\{\beta\})$ on all words of length $<n$, that are, all words defined by $u$.

    For the inductive case from $u_k$ to $u_{k+1}$, let $z=|u_k|$. By induction hypothesis $u_k\in\mathcal V$. It is sufficient to prove that there is a $y\in\{0,1\}$ such that $u_{k+1}\coloneqq u_ky$ is in $\mathcal V$, and satisfies the following two weaker properties (noting that, when moving from $u_k$ to $u_{k+1}$, the only word newly defined is $z$, and $z\in u_{k+1}$ if and only if $y=1$):
    \begin{gather*}
        \text{P1: }z\in Q_\alpha^+ \implies y=(u(\{\alpha\}))(z), \text{ and P2: } z\in Q_\beta^+ \implies y=(u(\{\beta\}))(z).
    \end{gather*}
    We perform this extension by one bit using Lemma~\ref{lemma:extension} with regard to $t_{\hat{s}-1}$-validity. Observe that cases (ii--iv) never apply.
    Assume otherwise, then there is some task $\tau\in\dom(t_{\hat{s}-1})$ with $|z|\in H_{m'}$ for $m'=t_{\hat{s}-1}(\tau)$.
    Note that $\tau\neq \tau^2_{i,j}$ since $\hat{s}$ is the first stage when $\tau^2_{i,j}$ is treated.
    Also, $|z|=n$ as otherwise $|z|\not\in H_{m'}$, by Observation~\ref{obs:leveldefinitions}(iii) and \eqref{eq:expbound};
    Since $t_{s}$ is an extension of $t_{s'}$ and is injective on its support, we know that $m'=t_{\hat{s}-1}(\tau)=t_{\hat{s}}(\tau)\neq t_{\hat{s}}(\tau^2_{i,j})=m$.
    But then $n\in H_m, H_{m'}$, which contradicts Observation~\ref{obs:leveldefinitions}(i) that $H_{m'}, H_m$ are disjoint.

    Now, as cases~\ref{lemma:extension}(ii--iv) never apply, the statement of the Lemma can be simplified as follows:
    \begin{equation}\label{eq:finalextension-helper}
        \parbox{14cm}{\setlength{\parskip}{1ex}If $z=c(a,b,x)$ and $0<t_{\hat{s}-1}(\tau^1_{a,b})\leq z$, then $u_k1\in\mathcal V$ when $M_a^{u_k}(x)$ accepts, and $u_k0\in\mathcal V$ when $M_a^{u_k}(x)$ rejects.

            Otherwise, $u_k0, u_k1\in\mathcal V$.}\tag{${\ast}{\ast}$}
    \end{equation}
    We now analyze three cases:

    Case 1: $z\not\in Q_\alpha^+\cup Q_\beta^+$. Here, \eqref{eq:finalextension-helper} asserts that for some suitable choice of $y$, $u_{k+1}=u_ky\in \mathcal V$. Conditions P1 and P2 are vacuously true.\medskip

    Case 2: $z\in Q_\alpha^+$. Let $y\coloneqq (u(\{\alpha\}))(z)$, thus P1 holds.
    Also, we know that P2 holds: either it is vacuously true, or,  when $z\in Q_\alpha^+\cap Q_\beta^+$ then, by absence of spoilage between $\alpha$ and $\beta$, $u(\{\alpha\})$ and $u(\{\beta\})$ agree on $z$.
    We can assume without loss of generality that $z=c(a,b,x)$ for suitable $a,b,x$, and $0<t_{\hat{s}-1}(\tau^1_{a,b})\leq z$. Otherwise \eqref{eq:finalextension-helper} asserts that $u_{k+1}=u_ky\in \mathcal V$ and we are done.

    Under this assumption, $z$ is an active code word. 
    We subdivide the remaining proof into subcases $y=1$ (i.e., $z\in u(\{\alpha\})$) and $y=0$ (i.e., $z\not\in u(\{\alpha\})$).
    In the first subcase, we have:
    \begin{equation*}
        y=1 \overset{\mathclap{\text{\ref{claim:welldefined}(iii)}}}\implies M_a^{{u(\{\alpha\})}}(x) \text{ acc.} \overset{\mathclap{\text{\ref{obs:partialoracles}(iii)}}}\implies M_a^{u_k}(x) \text{ acc.} \overset{\mathclap{\eqref{eq:finalextension-helper}}}\implies u_{k+1}=u_ky\in\mathcal V.
    \end{equation*}
    The first resp. last implication follows from Claim~\ref{claim:welldefined}(iii) resp. \eqref{eq:finalextension-helper}. It remains to prove the second one.
    Let $Q$ be the set of all oracle queries of the leftmost accepting path of $M_a^{\smash{u(\{\alpha\})}}(x)$. We prove that $u_k$ and $u(\{\alpha\})$ agree on $Q$:
    Let $q\in Q$. By definition, $q$ is a direct successor of $z\in Q^+_\alpha$ in $G({\{\alpha\}}$, hence also $q\in Q^+_\alpha$. Also, $q<z=|u_k|$, thus $q$ is defined by $u_k$.
    By induction hypothesis, it holds that $q\in u_k$ if and only if $q\in u(\{\alpha\})$, as desired.

    We have seen that $u_k$ and $u(\{\alpha\})$ agree on $Q$, thus Observation~\ref{obs:partialoracles}(iii) applies, and $M_a^{u_k}(x)$ accepts, as desired.

    The other subcase can be proven similarly:
    \begin{equation*}
        y=0 \overset{\mathclap{\text{\ref{claim:welldefined}(iii)}}}\implies M_b^{u(\{\alpha\})}(x) \text{ acc.}\overset{\mathclap{\text{\ref{obs:partialoracles}(iii)}}}\implies  M_b^{u_k}(x) \text{ acc.} \overset{\mathclap{\eqref{eq:finalextension-helper}}}\implies u_{k+1}=u_ky\in\mathcal V.
    \end{equation*}

    Case 3: $z\in Q_\beta^+$. Let $y\coloneqq (u(\{\beta\}))(z)$, and proceed similarly to Case 2.
\end{proof}

Remember that $t'=t_{\hat{s}-1}\cup\{\tau^2_{i,j} \mapsto 0\}$.
Oracle $u'$ from Claim~\ref{claim:finalextension-disjnp} is $t'$-valid:
As $u'\in\mathcal V$, it is $t_{\hat{s}-1}$-valid, hence only V3 with respect to $\tau^2_{i,j}$ is at risk.
However, with $x=F^{u'}_r(0^n)$, both $M_i^{u'}(x)$ and $M_j^{u'}(x)$ definitely accept, as stated by Claim~\ref{claim:finalextension-disjnp}.

We obtain the following situation: 
In stage $\hat{s}$ treating task $\tau^2_{i,j}$, oracle $u'$ is a possible $t'$-valid extension of $w_{\hat{s}-1}$, since it is $t'$-valid and $u'\sqsupsetneq w_{\hat{s}-1}$.
Thus, by definition of the task, we have that $t_{\hat{s}}=t'$, and thus $t_{s}(\tau^2_{i,j}) = t'(\tau^2_{i,j}) = 0$, contradicting the hypothesis of this Theorem~\ref{thm:disjnp}.

\begin{center}
    $\ast$
\end{center}

It remains to prove that there exists a pair $\alpha\in\Sigma^{n-1}0$, $\beta\in\Sigma^{n-1}1$ such that $\alpha$ and $\beta$ do not spoil each other.
For this, we set up a directed bipartite spoilage graph $S$ with left vertex part $A=\Sigma^{n-1}0$ and right vertex part $B=\Sigma^{n-1}1$.
We then give a definition of the edge set that captures (a stronger notion of) spoilage. That is, $\alpha$ and $\beta$ are connected in $S$ by an edge whenever $\alpha$ and $\beta$ spoil each other (meaning that $u(\{\alpha\})$ and $u(\{\beta\})$ do not agree on $Q^+_\alpha\cap Q^+_\beta$).
Formulated contrapositively, when $(\alpha, \beta), (\beta, \alpha)\not\in E(S)$, then $\alpha$ and $\beta$ do not spoil each other, as desired.

By \eqref{eq:expbound} both parts of the graph have $2^{n-1}>8\gamma(n)$ vertices.
We will bound the out-degree of all vertices by $\leq 4\gamma(n)$.
Thus, combinatorial Lemma~\ref{lemma:bipartite} applies, and there exist $\alpha\in \Sigma^{n-1}0$, $\beta\in\Sigma^{n-1}1$ such that neither $(\alpha,\beta)\in E(S)$ nor $(\beta,\alpha)\in E(S)$, as desired.

Define the edge set as follows: for every pair $\alpha\in\Sigma^{n-1}0$, $\beta\in\Sigma^{n-1}1$,
\begin{gather*}
    (\alpha, \beta) \in E(S) \iff \beta\in R_{G(\emptyset)}(Q^+_\alpha)\cap B,\\
    (\beta, \alpha) \in E(S) \iff \alpha\in R_{G(\emptyset)}(Q^+_\beta)\cap A
\end{gather*}
It remains to show two properties: first, the correctness that vertices without edges between them are in fact not spoiling each other, and second, the stated out-degree upper bound.

\begin{claim}
    Let $\alpha\in\Sigma^{n-1}0$, $\beta\in\Sigma^{n-1}1$.
    If $(\alpha, \beta)\not\in E(S), (\beta, \alpha)\not\in E(S)$, then  $\alpha$ and $\beta$ do not spoil each other:
    $u(\{\alpha\})$ and $u(\{\beta\})$ agree on $Q^+_\alpha\cap Q^+_\beta$.
\end{claim}
\begin{proof}
    By definition of $E(S)$, we have $\alpha\not\in R_{G(\emptyset)}(Q^+_\beta)$, and $\beta\not\in R_{G(\emptyset)}(Q^+_\alpha)$.
    Observe that we have $\alpha,\beta\not\in R_{G(\emptyset)}(Q^+_\alpha)\cap R_{G(\emptyset)}(Q^+_\beta)$.
    Denote with $G'$ the subgraph of $G(\emptyset)$ induced by the vertex set $R_{G(\emptyset)}(Q^+_\alpha)\cap R_{G(\emptyset)}(Q^+_\beta)$.
    We want to show that whenever $z\in Q^+_\alpha\cap Q^+_\beta\subseteq V(G')$, then $z\in u(\{\alpha\})$ if and only if $z\in u(\{\beta\})$.
    It suffices to show the following stronger statement. Let $\xi\in\{\alpha, \beta\}$, then:
    \begin{equation}\label{eq:spoilage-correctness-helper}
        \text{For all words $z\in V(G')$, it holds that}\quad z\in u(\emptyset) \iff z\in u(\{\xi\})\tag{${\ast}{\ast}{\ast}$}
    \end{equation}
    Then, for any $z\in Q^+_\alpha\cap Q^+_\beta$, above equivalences hold because $z\in V(G')$, and thus $z\in u(\{\alpha\})\iff z\in u(\emptyset)\iff z\in u(\{\beta\})$ as desired.

    Assume the equivalence does not hold, and choose $\xi\in\{\alpha, \beta\}$ and vertex $z\in V(G')$ with minimal height in $G'$ for which the equivalence \eqref{eq:spoilage-correctness-helper} fails.
    We analyze the three cases for which the membership of $z$ to the oracles can be defined.

    Suppose the membership of $z$ to $u(\emptyset), u(\{\xi\})$ is defined by clause (2), i.e., $|z|=n$. Then  $z\neq \alpha, \beta$ by above assertion that $\alpha, \beta\not\in V(G')$.
    This implies $z\not\in u(\emptyset),  u(\{\xi\})$, and \eqref{eq:spoilage-correctness-helper} proven.

    Suppose the membership of $z$ to $u(\emptyset), u(\{\xi\})$ is defined by clause (4), i.e., $z$ not an active code word. Then immediately by definition $z\not\in u(\emptyset), u(\{\xi\})$, and \eqref{eq:spoilage-correctness-helper} proven.

    For the remaining case (3), $z=c(a,b,x)$ is an active code word.
    We show that for all $z' \in N^+_{G(\emptyset)}(z)$, the equivalence from \eqref{eq:spoilage-correctness-helper} holds for $z'$.
    We have $z\in R_{G(\emptyset)}(Q^+_\alpha)$ by definition, hence also $z'\in R_{G(\emptyset)}(Q^+_\alpha)$. Similarly, we also have $z'\in R_{G(\emptyset)}(Q^+_\beta)$; hence we have $z'\in R_{G(\emptyset)}(Q^+_\alpha) \cap R_{G(\emptyset)}(Q^+_\beta) = V(G')$.
    Also, $z'$ has smaller height than $z$ in $G'$, hence by choice of $z$, the statement \eqref{eq:spoilage-correctness-helper} holds for $z'$, and $z'\in u(\emptyset)$ if and only if $z'\in u(\{\xi\})$.
    In total, $u(\emptyset)$ and $u(\{\xi\})$ agree on $N^+_{G(\emptyset)}(z)$.

    Now say without loss of generality that the equivalence \eqref{eq:spoilage-correctness-helper} fails because $z\not\in u(\emptyset)$ but $z\in u(\{\xi\})$. (The other case holds symmetric – replace $b$ with $a$.)
    Remember that $z$ is an active code word in $G(\emptyset)$ and $G(\{\xi\})$, thus by Claim~\ref{claim:welldefined}(iii), this means that $M_b^{{u(\emptyset)}}(x)$ accepts and $M_b^{{u(\{\xi\})}}(x)$ rejects.

    We have seen that $u(\emptyset)$ and $u(\{\xi\})$ agree on $N^+_{G(\emptyset)}(z)$, which by clause (3) contains the set of oracle queries of the leftmost accepting path of $M_b^{u(\emptyset)}(x)$. Invoking Observation~\ref{obs:partialoracles}(iii) we obtain that also $M_b^{\smash{u(\{\xi\})}}(x)$ accepts.
    This contradicts the previous assertion that $M_b^{\smash{u(\{\xi\})}}(x)$ rejects.
\end{proof}

For the out-degree bound, consider the following claim.

\begin{claim}\label{claim:frozenwords-len}
    Let $X\subseteq\Sigma^n$, $|X|\leq 1$, $Q\subseteq\Sigma^{\leq \gamma(n)}$.

    Then $\ell(R_{G({X})}(Q)) \leq 2\cdot \ell(Q)$.
\end{claim}
\begin{proof}
    Consider the graph $G'$ that is formed by taking the subgraph from $G({X})$ induced by $R_{G({X})}(Q)$.
    We show inductively on the height of $z$ in $G'$ that
    \begin{equation*}
        \ell(R_{G'}(z))\leq 2|z| \text{ for all $z\in V(G')$.}
    \end{equation*}
    Then,
    \begin{equation*}
        \ell(R_{G({X})}(Q)) \leq \sum_{q\in Q} \ell(R_{G'}(q)) \leq \sum_{q\in Q} 2|q| = 2\cdot \ell(Q),
    \end{equation*}
    as desired.

    The base case with height $0$ is immediate, since then $|R_{G'}(z)|=1$.
    For the inductive case, we know that $\mathrm{outdeg}_{G'}(z)>0$, hence by construction of $G({X})$, there exist $x,a,b$ such that $z=c(a, b, x)$,
    and the word $z$ is an active code word.
    By induction hypothesis, $\ell(R_{G'}(q))\leq 2|q|$ for all $q\in N^+_{G'}(z)$ since these have lower height.

    Without loss of generality, $z\in u(X)$ and $M_a^{{u(X)}}(x)$ accepts.
    (In the other case, $z\not\in u(X)$ and $M_b^{{u(X)}}(x)$ accepts.)
    Note that by Claim~\ref{claim:welldefined}(iii), we know that $M_b^{{u(X)}}(x)$ rejects.
    Let $Q_X$ be the set of oracle queries on the leftmost accepting path of $M_a^{{u(X)}}(x)$, that is, precisely as selected in clause (3). By construction, $N^+_{G'}(z) = Q_X$, as only $M_a^{{u(X)}}(x)$ accepts.
    By Claim~\ref{claim:codewords}(iii), $\ell(Q_X)\leq |z|/2$. (We use a weaker bound than actually claimed.)
    Hence $\ell(N^+_{G'}(z))= \ell(Q_X)\leq |z|/2$, and thus,
    \begin{align*}
        \ell(R_{G'}(z)) &= |z| + \sum_{q\in N^+_{G'}(z)} \ell(R_{G'}(q)) \leq |z| + \sum_{q\in N^+_{G'}(z)} 2|q|\quad\text{(by induction hypothesis)}\\
                        &= |z| + 2\ell(N^+_{G'}(z)) \leq |z|+|z| = 2|z|.\qedhere
    \end{align*}
\end{proof}

By definition, the successors of some vertex $\xi$ in $S$ are $ N^+_S(\xi) \subseteq R_{G(\emptyset)}(Q^+_\xi)$,
hence
\[ |N^+_S(\xi)| \leq|R_{G(\emptyset)}(Q^+_\xi)|\leq \ell(R_{G(\emptyset)}(Q^+_\xi)). \]
(Remember that $Q_\xi$ was defined as the queries on the leftmost accepting path of $M_i(F_r(0^n))$ relative to $u(\{\xi\})$ when $\xi$ ends with $0$, and the queries on the leftmost accepting path of $M_j(F_r(0^n))$ otherwise.)
Now, with initial observation $\ell(Q_\xi) \leq \gamma(n)$ at the beginning of this proof,
\[ |N^+_S(\xi)| \leq \ell(R_{G(\emptyset)}(Q^+_\xi)) = \ell(R_{G(\emptyset)}(R_{G(\{\xi\})}(Q_\xi))) \leq 2\ell(R_{G(\{\xi\})}(Q_\xi)) \leq 2(2\ell(Q_\xi)) \leq 4\gamma(n), \]
as was desired.

This completes the proof of Theorem~\ref{thm:disjnp}.
\end{proof}

For the tasks $\tau^3_{i,r}$, we only need to slightly modify the previous proof.

\begin{theorem}\label{thm:up}
    Let $s\in \mathbb N^+$, $(w_0, t_0), \dots, (w_{s-1}, t_{s-1})$ defined.
    Consider task $\tau^3_{i,r}$. 

    Suppose that $t_s=t_{s-1}$, $t_s(\tau^3_{i})=m>0$.
    Then there exists a $t_{s}$-valid $w\sqsupsetneq w_{s-1}$ and $n\in\mathbb N$ such that one of the following holds:
    \begin{enumerate}[noitemsep]
        \item $0^n\in C_m^v$ for all $v\sqsupseteq w$ and $M_i^{w}(F_r^{w}(0^n))$ definitely rejects.
        \item $0^n\not\in C_m^v$ for all $v\sqsupseteq w$ and $M_i^{w}(F_r^{w}(0^n))$ definitely accepts.
    \end{enumerate}
\end{theorem}
\begin{proof}
Proceed just as in the proof of Theorem~\ref{thm:disjnp}, assuming $i=j$.
Aiming for a contradiction, we assume that for all $t_s$-valid $w\sqsupsetneq w_{s-1}$, (i) and (ii) do not hold.
We thus have:
\begin{itemize}[noitemsep]
    \item $M_i(F_r(0^n))$ rejects relative to $u(\emptyset)$, 
    \item Whenever $\alpha\in\Sigma^n$ ends with $0$, $M_i(F_r(0^n))$ accepts relative to $u(\{\alpha\})$, 
    \item Whenever $\beta\in\Sigma^n$ ends with $1$, $M_i(F_r(0^n))$ accepts relative to $u(\{\beta\})$.
\end{itemize}

Similar to previous proof, let $Q_\alpha$ (resp., $Q_\beta$) be, respectively, the oracle queries of leftmost accepting path of $M_i(F_r(0^n))$ relative to $u(\{\alpha\})$ (resp., relative to $u(\{\beta\})$), and let $Q^+_\alpha\coloneqq R_{G({\{\alpha\}})}(Q_\alpha)$ be $Q_\alpha$'s closure under successor in $G({\{\alpha\}})$. Define $Q^+_\beta$ respectively.
Observe that $\alpha\in Q_\alpha$, otherwise $u(\emptyset)$ and $u(\{\alpha\})$ agree on $Q_\alpha$, Observation~\ref{obs:partialoracles}(iii) applies, and $M_i(F_r(0^n))$ rejects relative to $u(\emptyset)$, contradicting the assumption.
Similarly, $\beta\in Q_\beta$.

Proceed, as in the proof of Theorem~\ref{thm:disjnp}, fixing $\alpha\in\Sigma^{n-1}0$, $\beta\in\Sigma^{n-1}1$ that do not spoil each other. That is, $u(\{\alpha\})$ and $u(\{\beta\})$ agree on $Q^+_\alpha \cap Q^+_\beta$.
By same reasoning, Claim~\ref{claim:finalextension-disjnp} applies, and thus there is a $t_{\hat{s}-1}$-valid $u'$ defined for all words of length $\leq \gamma(n)$, and
\begin{itemize}[noitemsep]
    \item $u'$ agrees with $u(\{\alpha\})$ on $Q_\alpha$, and
    \item $u'$ agrees with $u(\{\beta\})$ on $Q_\beta$.
\end{itemize} 
With Observation~\ref{obs:partialoracles}(ii) and (iii), we obtain that $M_i(F_r(0^n))$ definitely accepts relative to $u'$, whereas one accepting computation path queries set $Q_\alpha$, and another accepting computation path queries set $Q_\beta$.
These computation paths are different: assume otherwise, then $Q_\alpha=Q_\beta$ and we obtain $\alpha\in Q_\alpha=Q_\beta$.
This means that $\alpha \in Q_\alpha \cap Q_\beta\subseteq Q^+_\alpha \cap Q^+_\beta$, but $\alpha\in u(\{\alpha\})$, $\alpha\not\in u(\{\beta\})$, hence $\alpha$ and $\beta$ spoil each other.
This contradicts the choice of $\alpha, \beta$.

Assign $t'\coloneqq t_{\hat{s}-1}\cup\{\tau^3_{i} \mapsto 0\}$.
Oracle $u'$ is $t'$-valid: it is $t_{\hat{s}-1}$-valid, hence only V3 with respect to $\tau^3_{i}$ is at risk.
However, with $x=F^{u'}_r(0^n)$, we already have seen that $M_i^{u'}(x)$ definitely accepts on two different paths.

We obtain following situation: 
In stage $\hat{s}$ treating task $\tau^3_{i}$, oracle $u'$ is a possible $t'$-valid extension of $w_{\hat{s}-1}$ since it is $t'$-valid and $u'\sqsupsetneq w_{\hat{s}-1}$.
Thus, by definition of the task, we have that $t_{\hat{s}}=t'$, and thus $t_{s}(\tau^3_{i}) = t'(\tau^3_{i}) = 0$, contradicting the hypothesis of this Theorem~\ref{thm:up}.
\end{proof}

We now show that the construction for task $\tau^4_{i,j,r}$ is possible.
\begin{theorem}\label{thm:disjconp}
    Let $s\in \mathbb N^+$, $(w_{s-1}, t_{s-1})$ defined.
    Consider task $\tau^4_{i,j,r}$. 

    Suppose that $t_s=t_{s-1}$, $t_{s}(\tau^4_{i,j})=m>0$.
    Then there exists a $t_{s}$-valid $w\sqsupsetneq w_{s-1}$ and $n\in\mathbb N$ such that one of the following holds:
    \begin{enumerate}[noitemsep]
        \item $0^n\in D_m^v$ for all $v\sqsupseteq w$ and $M_i^{w}(F_r^{w}(0^n))$ definitely accepts.
        \item $0^n\in E_m^v$ for all $v\sqsupseteq w$ and $M_j^{w}(F_r^{w}(0^n))$ definitely accepts.
    \end{enumerate}
\end{theorem}
\begin{proof}
Let us fix  $i,j,r$ throughout the proof of the theorem.

Let $\hat{s}<s$ be the stage that treated $\tau^4_{i,j}$.
Such stage exists, as otherwise $t_{s}(\tau^4_{i,j})$ is undefined.
We have $m=t_{\hat{s}}(\tau^4_{i,j})=t_s(\tau^4_{i,j})$; fix $m$ for the rest of the proof.
Unlike the proof by contradiction used for Theorem~\ref{thm:disjnp}, for this theorem we explicitly construct a $t_{s}$-valid $w\sqsupsetneq w_{s-1}$ such that one of (i) or (ii) holds.

Again, let
\begin{equation*} \gamma(n) \coloneqq  \max(p_i(p_r(n))+p_r(n), p_j(p_r(n))+p_r(n)) \end{equation*}
be the polynomial bounding the runtime of $M_i\circ F_r$, $M_j\circ F_r$ with respect to input length $n$ relative to any oracle.
Let us define $n\in\mathbb N^+$ as the smallest $n\in H_m$ such that $w_{s-1}$ does not define any words of length $\geq n$,
and
\begin{equation}\label{eq:expbound-disjconp}
    2^n > \gamma(n),\quad 2^n/2 > 4\gamma(n),\quad 2^n/4 > \gamma(n)/(2n).\tag{${\ast}{\ast}{\ast}{\ast}$}
\end{equation}
Again, first inequality of \eqref{eq:expbound-disjconp} ensures that no level $n<n'\leq \gamma(n)$ is reserved for any witness language, that is $n'\not\in H_0, H_1, \dots$ (cf. Observation~\ref{obs:leveldefinitions}(iii)). The following inequalities ensure that there are enough words of length $n$ such that the combinatorial arguments work, which we will employ later.
For the remaining proof, we additionally fix $n$.
Again, we define $u\sqsupseteq w_{s-1}$ as the minimal $t_{s}$-valid partial oracle which is defined precisely for all words up to length $<n$. Such oracle exists by Lemma~\ref{lemma:extension}, by extending $w_{s-1}$ bitwise such that it remains $t_{s}$-valid.

For $X \subseteq \Sigma^n$, we will define $u(X)\sqsupsetneq u$, which is defined for all words of length $\leq\gamma(n)$, and such that $u(X) \cap \Sigma^n = X$. 
To each $u(X)$, we define a directed graph $G(X)$ on the vertex set $\Sigma^{\leq\gamma(n)}$.
The definitions are very similar to the definitions in the proof of Theorem~\ref{thm:disjnp}. For all $X$ with $|X|=1$ and certain $X$ with $|X|=2$, we can attain a $t_s$-valid $u(X)$ because $X$ is ``compatible'' with V8.

We also handle in our definition the specific case $X=\emptyset$.
Due to ``incompatibility'' with V8, $u(\emptyset)$ is not $t_s$-valid, but we define $u(\emptyset)$ ``as $t_s$-valid as possible'', meaning that we can extend $u(\emptyset)$ to some $u(\{\xi\})$ for $\xi \in \Sigma ^n$ while fixing some important words.

In fact, V8 requires that for some $X\subseteq \Sigma^n$, $|X|\geq 1$ and that all words in $X$ have pairwise different parity. In other words, either $X=\{\alpha\}$, or $X=\{\beta\}$ or $X=\{\alpha,\beta\}$ for some $\alpha\in\Sigma^{n-1}0, \beta\in\Sigma^{n-1}1$. For brevity, we write that $X$ is \emph{V8-compatible} if $X\subseteq\Sigma^n$ satisfies this property. Accordingly, the empty set $\emptyset$ is V8-incompatible. We design our constructions such that $u(X)$ is $t_s$-valid if $X$ is V8-compatible.

\paragraph{Definition of $u(X)$, $G(X)$:} Let $X\subseteq\Sigma^n$. We construct $u(X)$ and $G(X)=(V,E)$ inductively. Fix vertex set $V=\Sigma^{\leq\gamma(n)}$.
Basis clauses:
\begin{enumerate}[label=(\arabic*)]
    \item For $z\in \Sigma^{<n}$, set $z\in u(X)$ if and only if $z\in u$.
    \item For $z\in \Sigma^{n}$, set $z\in u(X)$ if and only if $z\in X$.
\end{enumerate}
Inductive clauses: Let $z\in\Sigma^{\leq \gamma(n)}$, $|z|>n$, and $u(X)$ defined for words $<z$.
\begin{enumerate}[label=(\arabic*),resume*]
    \item If $z=c(a,b,x)$ for suitable $a,b,x$ with $0<t_s(\tau^1_{a,b})\leq z$, and \emph{at least one} of $M^{{u(X)}}_{a}(x)$ or $M_{b}^{{u(X)}}(x)$ accepts, continue as follows:

        Mark vertex $z$ as \emph{correct code word}.
        If $M_a^{{u(X)}}(x)$ accepts, then define $z\in u(X)$.
        Let $(z,q)\in E$ for all oracle queries $q$ on the leftmost accepting path of $M_a^{u(X)}(x)$.

        Otherwise, $M_b^{{u(X)}}(x)$ accepts, and define $z\not\in u(X)$.
        Let $(z,q)\in E$ for all oracle queries $q$ on the leftmost accepting path of $M_b^{u(X)}(x)$.

        (Meaning: If $z$ is a code word for the codings of $\tau ^1_{a,b}$ and one of $M_a^{{u(X)}}(x)$ and $M_b^{{u(X)}}(x)$ has an accepting path, we can construct ${u(X)}$ for $z$ like in Lemma~\ref{lemma:extension}, i.e., the coding is correct with respect to Lemma~\ref{lemma:extension}. The dependencies of this accepting path on other words in ${u(X)}$ are captured by adding respective edges to the edge set of $G(X)$.)

    \item Otherwise, if $z=c(a,b,x)$ for suitable $a,b,x$ with $0<t_s(\tau^1_{a,b})\leq z$, and \emph{none} of $M^{{u(X)}}_{a}(x)$ or $M_{b}^{{u(X)}}(x)$ accepts, continue as follows:

        Mark vertex $z$ as \emph{incorrect code word}.
        Define $Q^z_\mathrm{in} \coloneqq \{ \xi\in\Sigma^n \mid z\in u(\{\xi\}) \}$, $Q^z_\mathrm{out} \coloneqq \{ \xi\in\Sigma^n \mid z\not\in u(\{\xi\})\}$.
        If $|Q^z_\mathrm{out}|\leq |Q^z_\mathrm{in}|$, let $z\in u(X)$ and $(z,\xi)\in E$ for all $\xi\in Q^z_\mathrm{out}$.
        Otherwise, define $z\not\in u(X)$ and let $(z,\xi)\in E$ for all $\xi\in Q^z_\mathrm{in}$.

        (Meaning: The machines $M_a$ and $M_b$ do not work complementary relative to ${u(X)}$, which can happen for $X$ that are V8-incompatible, because V8 is violated, and therefore ${u(X)}$ is not $t_s$-valid. Since we later want to add a word $\xi$ of length $n$ with specific properties, we now determine the membership of $z$ to $u(X)$ according to the majority of possible one-word-extensions of $X$. The dependencies on words that induce a different behavior for $z$ are captured by adding respective edges to the edge set of $G(X)$.
        Later we see that the number of added edges is small, as either $|Q^z_\mathrm{in}|\gg|Q^z_\mathrm{out}|$ or $|Q^z_\mathrm{in}|\ll|Q^z_\mathrm{out}|$.)

    \item Otherwise, $z\not\in u(X)$.
\end{enumerate}
Extremal clause: (6) No other edges are in $E$.

Let us repeat the following definitions used in clause (4):
\begin{equation*}
    Q_\mathrm{in}^z\coloneqq \{\xi\in\Sigma^n \mid z\in u(\{\xi\})\},\quad Q_\mathrm{out}^z\coloneqq \{\xi\in\Sigma^n \mid z\not\in u(\{\xi\})\},
\end{equation*}
and observe that $Q_\mathrm{in}^z, Q_\mathrm{out}^z$ partition $\Sigma^n$.

We make some claims concerning $u(X),G(X)$.
\begin{claim}\phantomsection\label{claim:welldefined-disjconp}
    \begin{enumerate}
        \item Whenever $X$ is V8-compatible, $u(X)$ is well-defined, is $t_s$-valid and $G(X)$ contains no incorrect code words. (But may contain words of the form $c(\cdot,\cdot,\cdot)$ that are not marked as correct code words.)
        \item For any arbitrary $X\subseteq \Sigma^n$, $u(X)$ is well-defined for all words of length $\leq \gamma(n)$, $u(X) \cap \Sigma^n = X$, and $u(X) \sqsupsetneq u\sqsupseteq w_{s-1}$.

        \item Whenever $X$ is V8-compatible, then for every correct code word $z=c(a,b,x)$ in $G(X)$, the following statements are equivalent: (a) $z\in u(X)$, (b) $M_a^{u(X)}(x)$ accepts, (c) $M_b^{u(X)}(x)$ rejects.
        \item Whenever $s'<\hat{s}$, $u(\emptyset)$ is $t_{s'}$-valid.
        \item Let $z=c(a,b,x)$ be a correct code word in $G(X)$.
            Then $z\in u(X)$ implies $M_a^{u(X)}$ accepts, and $z\not\in u(X)$ implies $M_b^{u(X)}$ accepts.
        \item $G(X)$ forms a directed acyclic graph (which is not necessarily connected).
            In particular, for every directed edge from vertex $a$ to $b$, it holds that $a>b$.
    \end{enumerate}
\end{claim}
\begin{proof}
    To (i): Let $u_0, u_1, u_2, \dots, u_l$ be a length-ordered enumeration of all prefixes of $u(X)$ that are defined for at least all words of length $<n$, that is, $u = u_0 \sqsubsetneq u_1 \sqsubsetneq u_2 \sqsubsetneq \dots \sqsubsetneq u_l=u(X)$.
    Observe that by definition, for every $u_k$, $0<k\leq l$, the last bit of $u_k$ is defined by one of the clauses (2)--(5).
    We show inductively that every $u_k$ is well-defined, is $t_s$-valid and the last bit of $u_k$ is not defined by clause (4).
    In total, we obtain that the claimed assertions hold for $u_l=u(X)$, and specifically $G(X)$ does not contain incorrect code words as these are only marked in clause (4).

    Base case is immediate, as $u_0 = u$ is well-defined by exclusive definitions from clause (1) and is $t_s$-valid by choice.

    For the inductive case from $u_k$ to $u_{k+1}$, let $z=|u_k|$, and $y=(u(X))(z)$. It holds that $u_ky=u_{k+1}\sqsubseteq u(X)$, i.e., $y$ is the last bit of $u_{k+1}$.
    By induction hypothesis, $u_k$ is well-defined and is $t_s$-valid.
    Note that $y$ is defined by above inductive definition by one of the clauses (2)--(5).

    First we see that $y$ is not defined by clause (4): assume otherwise, then $z=c(a,b,x)$ for suitable $a,b,x$ and we have $0<t_{s}(\tau^1_{a,b})\leq z$.
    By hypothesis $u_k$ is $t_s$-valid, hence conditions of Lemma~\ref{lemma:npconp-tasks} apply and precisely one of $M_a^{u_k}(x)$ and $M_b^{u_k}(x)$ accepts. 
    Say without loss of generality that $M_a^{u_k}(x)$ accepts. 
    By Claim~\ref{claim:codewords}(iv) and Observation~\ref{obs:partialoracles}(i), also $M_a^{u(X)}(x)$ accepts.
    This contradicts the definition from clause (4).

    Thus, $y$ is defined by either clause (2), (3) or (5). Bit $y$ is well-defined, in the sense that $y$ is entirely defined by the partial oracle $u_{k}$:
    Either $y$ is defined by clause (2) or (5) and is trivially well-defined, or $y$ is defined by clause (3), that is $z=c(a,b,x)$ for suitable $a,b,x$.
    In this case $M_a^{u(X)}(x)$ and $M_b^{u(X)}(x)$ are definite by Claim~\ref{claim:codewords}(iv), hence can only ask queries $<z=|u_k|$; thus $y$ is well-defined.

    It remains to show that $u_ky=u_{k+1}$ is $t_s$-valid. For this, we will employ Lemma~\ref{lemma:extension} with respect to $t_s$-valid $u_k$.
    We analyze three cases on which clause (2), (3) or (5) defines $y$.
    \medskip

    Clause (3): We have $z=c(a,b,x)$ for suitable $a,b,x$, and case~\ref{lemma:extension}(i) applies.
    Since we have $y=1$ if and only if $M_a^w(x)$ accepts, we obtain that $u_ky$ is $t_s$-valid.\medskip

    Clause (5): Note that $n<|z|<2^n$ by \eqref{eq:expbound}, hence by Observation~\ref{obs:leveldefinitions}(iii) the cases~\ref{lemma:extension}(ii--iv) cannot apply.
    Furthermore, case~\ref{lemma:extension}(i) cannot apply either. Otherwise we have $z= c(a,b,x)$ for suitable $a,b,x$ and $0<t_{s}(\tau^1_{a,b})\leq z$. However, this case would be handled in clause (3).
    In total, only case~\ref{lemma:extension}(v) applies, hence $u_ky$ is $t_s$-valid.\medskip

    Clause (2):
    Observe that by definition of clause (2), $u_ky$ and $X$ agree on all words of length $n$ that are defined by $u_ky$.

    The case~\ref{lemma:extension}(i) cannot apply, since otherwise we have $z= c(\cdot,\cdot,\cdot)$ but by clause (2),  $|z|=n\in H_m$, contradicting Claim~\ref{claim:codewords}(i).
    The case~\ref{lemma:extension}(ii) cannot apply, since otherwise $t_s(\tau)=m'>0$ for some $\tau\neq \tau^4_{i,j}$ and $|z|=n\in H_{m'}$. Since $t_s$ is injective on its support, $m'=t_s(\tau)\neq t_s(\tau^4_{i,j})=m$. But then $n\in H_m, H_{m'}$, which contradicts Observation~\ref{obs:leveldefinitions}(i) that $H_m, H_{m'}$ are disjoint.

    Assume that~\ref{lemma:extension}(iv) applies.
    From the conditions of the case follows that $u_k\cap\Sigma^n=\emptyset$ and $u_ky=u_{k+1}$ is defined for all words of length $n$.
    This means that $y=1$ as otherwise $u_k0\cap\Sigma^n=\emptyset$ and by above observation $X=\emptyset$, contradicting the condition of this Claim~\ref{claim:welldefined-disjconp}(ii) that $X$ is V8-compatible.
    As $y=1$, case~\ref{lemma:extension}(iv) asserts that $u_k1=u_{k+1}$ is $t_s$-valid.

    Assume that~\ref{lemma:extension}(iii) applies.
    By the conditions of the case, there is some $x\in u_k\cap\Sigma^n$ that has the same parity as $z$. Note that $x<|u_k|=z$ thus $x$ and $z$ are different.
    This means that $y=0$ as otherwise $x,z\in u_k1\cap\Sigma^n$, and by previous observation, $X$ contains two different words of the same parity.
    This contradicts the condition of this Claim~\ref{claim:welldefined-disjconp}(i) that $X$ is V8-compatible.
    As $y=0$, case~\ref{lemma:extension}(iii) asserts that $u_k0=u_{k+1}$ is $t_s$-valid.

    Otherwise, if cases~\ref{lemma:extension}(i--iv) do not apply, case~\ref{lemma:extension}(v) asserts that $u_ky=u_{k+1}$ is $t_s$-valid.
    \medskip

    To (ii):
    By previous assertion, $u(X)$ is well-defined for V8-compatible $X$. For any other V8-incompatible $X\subseteq\Sigma^n$, $u(X)$ is also well-defined. In particular, the oracles $u(\{\xi\})$ invoked in clause (4) are completely well-defined by this Claim~\ref{claim:welldefined-disjconp}(i).

    The remaining assertions immediately follow from definition.
    \medskip

    To (iii):
    The equivalence of (a) and (b) immediately follows from the definition in clause (3). We prove that (b) and (c) are equivalent.
    From the previous Claim~\ref{claim:welldefined-disjconp}(i), it follows that $u(X)$ is $t_s$-valid, hence $t_{s-1}$-valid. 
    Also, previous Claim~\ref{claim:welldefined-disjconp}(ii) shows $u(X)\sqsupseteq w_{s-1}$.

    The conditions of Lemma~\ref{lemma:npconp-tasks} (invoked with regard to $t_{s-1}$-valid $u(X)\sqsupseteq w_{s-1}$) apply: We have $0<t_{s-1}(\tau^1_{a,b})=t_{s}(\tau^1_{a,b})\leq z\leq |u(X)|$ by definition of clause (3).
    Hence $M_a^{u(X)}(x)$ accepts if and only if $M_b^{u(X)}(x)$ rejects.
    \medskip

    To (iv):
    Like the above proof of~\ref{claim:welldefined-disjconp}(i).
    Let $u_0, u_1, u_2, \dots, u_l$ be a length-ordered enumeration of all prefixes of $u(\emptyset)$ that are defined for at least all words of length $<n$, that is, $u = u_0 \sqsubsetneq u_1 \sqsubsetneq u_2 \sqsubsetneq \dots \sqsubsetneq u_l = u(\emptyset)$.
    We show inductively that every $u_k$ is $t_{s'}$-valid.
    Thus, we obtain that $u_l=u(\emptyset)$ is $t_{s'}$-valid. 
    Base case is immediate, as $u_0 = u$ is $t_s$-valid by choice.

    For the inductive case from $u_k$ to $u_{k+1}$, let $z=|u_k|$, and $y=(u(X))(z)$. It holds that $u_ky=u_{k+1}\sqsubseteq u(X)$, i.e., $y$ is the last bit of $u_{k+1}$.
    By induction hypothesis, $u_k$ is $t_{s'}$-valid.
    Note that $y$ is defined by above inductive definition by one of the clauses (2)--(4).

    We need to show that $u_ky=u_{k+1}$ is $t_{s'}$-valid. For this, we will employ Lemma~\ref{lemma:extension} with respect to $t_{s'}$-valid $u_k$.
    We analyze four cases on which clause (2), (3), (4) or (5) defines $y$.\medskip

    Clause (2):
    The case~\ref{lemma:extension}(i) cannot apply, since otherwise we have $z= c(\cdot,\cdot,\cdot)$ but by clause (2),  $|z|=n\in H_m$, contradicting Claim~\ref{claim:codewords}(i).
    The cases~\ref{lemma:extension}(ii--iv) cannot apply, since otherwise $|z|=n\in H_{m'}$ for some $\tau$ and $m'=t_{s'}(\tau)$.
    As $s'<\hat{s}$ where $\tau^4_{i,j}$ is treated, we know that $\tau\neq \tau^4_{i,j}$.
    Since $t_{s}$ is an extension of $t_{s'}$ and is injective on its support, we know that $m'=t_{s'}(\tau)=t_{s}(\tau)\neq t_{s}(\tau^4_{i,j})=m$.
    But then $n\in H_m, H_{m'}$, which contradicts Observation~\ref{obs:leveldefinitions}(i) that $H_{m'}, H_m$ are disjoint.

    Hence, case~\ref{lemma:extension}(v) asserts that $u_ky=u_{k+1}$ is $t_s$-valid.\medskip

    Clause (4): Under the conditions of this case, $z=c(a,b,x)$ for suitable $a,b,x$ and $0<t_s(\tau^1_{a,b})$, where both $M_a(x)$ and $M_b(x)$ reject relative to $u(\emptyset)$.

    We show that $t_{s'}(\tau^1_{a,b})$ is either undefined or $>z$.
    Hence the case~\ref{lemma:extension}(i) cannot apply.
    By the same argument as above, the cases~\ref{lemma:extension}(ii--iv) cannot apply either.
    Thus case~\ref{lemma:extension}(v) asserts that $u_ky=u_{k+1}$ is $t_s$-valid.

    Therefore, assume $t_{s'}(\tau^1_{a,b})\leq z$ is defined.
    Since $t_s$ is an extension of $t_{s'}$, we have that $0<t_{s'}(\tau^1_{a,b})\leq z$.
    Now, as $u_k$ is $t_{s'}$-valid by induction hypothesis, the conditions of Lemma~\ref{lemma:npconp-tasks} are met, and one of $M_a^{u_k}(x)$ or $M_b^{u_k}(x)$ definitely accepts.
    Without loss, $M_a^{u_k}(x)$ accepts, and with Observation~\ref{obs:partialoracles}(i), $M_a^{u(\emptyset)}(x)$ accepts. This contradicts the conditions of the clause (4).
    \medskip

    If $z$ is defined by the other clauses (3) or (5), then we obtain a contradiction identical to the proof of~\ref{claim:welldefined-disjconp}(i).
    \medskip

    To (v): Immediately follows from definition in clause (3).
    \medskip

    To (vi): Edges are only added in clauses (3) and (4).
    For edges $(z,q)$ added in clause (3), Claim~\ref{claim:codewords}(iii) asserts that $q<z$.
    For edges $(z,\xi)$ added in clause (3), observe that $|\xi|=n<|z|$ hence $\xi<z$.
\end{proof}
Our new definition of $G(X)$ for V8-compatible $X$ admits the same bounds as our previous definition when looking at $\ell(R_{G(X)}(\cdot ))$.
\begin{claim}\label{claim:frozenwords-len-simple}
    Let $X\subseteq\Sigma^n$, $X$ V8-compatible, and $Q\subseteq\Sigma^{\leq \gamma(n)}$.

    Then $\ell(R_{G({X})}(Q)) \leq 2\cdot \ell(Q)$.
\end{claim}
\begin{proof}
    This can be proven inductively similar to the proof of Claim~\ref{claim:frozenwords-len}.
    Consider the graph $G'$ that is formed by taking the subgraph from $G(X)$ induced by $R_{G(X)}(Q)$.
    We show inductively on the height of $z$ in $G'$ that
    \begin{equation*} \ell(R_{G'}(z))\leq 2|z| \text{ for all $z\in V(G')$.} \end{equation*}
    This is sufficient to prove the asserted claim.

    The base case with height $0$ is immediate.
    For the inductive case, we know that $\mathrm{outdeg}_{G'}(z)>0$, hence by construction of $G({X})$, there exist $x,a,b$ such that $z=c(a, b, x)$.
    Also, word $z$ is a \emph{correct} code word, as by Claim~\ref{claim:welldefined-disjconp}, no incorrect words exist in $G(X)$.
    By induction hypothesis, $\ell(R_{G'}(q))\leq 2|q|$ for all $q\in N^+_{G'}(z)$.

    Without loss of generality, $M_a^{{u(X)}}(x)$ accepts.
    (In the other case, $M_b^{{u(X)}}(x)$ accepts.)
    Let $Q_X$ be the set of oracle queries on the leftmost accepting path.
    By Claim~\ref{claim:codewords}(iii), $\ell(Q_X)\leq |z|/2$. (We use a weaker bound than actually claimed.)
    By construction, $N^+_{G'}(z) = Q_X$, hence $\ell(N^+_{G'}(z))\leq |z|/2$.
    Then,
    \begin{align*}
        \ell(R_{G'}(z)) &= |z| + \sum_{q\in N^+_{G'}(z)} \ell(R_{G'}(q))) \leq |z| + \sum_{q\in N^+_{G'}(z)} 2|q|\quad\text{(by induction hypothesis)}\\
                        &= |z| + 2\ell(N^+_{G'}(z)) \leq |z|+|z| = 2|z|.\qedhere
    \end{align*}
\end{proof}
The next claim shows, that if we want to make sure that two oracles $u(\{\xi\})$ and $u(X)$ for V8-compatible $X \subseteq \Sigma ^n$ agree on a set $Q \subseteq \Sigma ^{\leq \gamma (n)}$ of oracle queries, we only have to prove that they agree on the words of length $n$ in $R_{G(\{\xi\})}(Q)$. Intuitively, the set $R_{G(\{\xi\})}(Q)$ contains the dependencies the words in $Q$ have on smaller words and $u(\{\xi\})$ and $u(X)$ always agree on words $<n$. If now $u(\{\xi\})$ and $u(X)$ also behave the same for words of length $n$ in $R_{G(\{\xi\})}(Q)$, the words in $Q$ have to behave equally in $u(\{\xi\})$ and $u(X)$.
\begin{claim}\label{claim:freezing-disjconp-simple}
    Let $\xi\in\Sigma^n, X\subseteq \Sigma^n$, $X$ V8-compatible, and $Q\subseteq\Sigma^{\leq\gamma(n)}$.

    Suppose that $\{\xi\}$ and $X$ agree on $R_{G(\{\xi\})}(Q) \cap \Sigma^n$. 
    Then $u(\{\xi\})$ and $u(X)$ agree on $Q$. 
\end{claim}
\begin{proof}
    We prove the stronger statement, that $u(\{\xi\})$ and $u(X)$ agree on $R_{G(\{\xi\})}(Q)\supseteq Q$.
    Consider graph $G'$ that is formed by taking the subgraph from $G(\{\xi\})$ induced by $R_{G({\{\xi\}})}(Q)$.
    We show that:
    \begin{equation*}  \text{For all words $z\in V(G')$, it holds that}\quad z\in u(\{\xi\})\iff z\in u(X). \end{equation*}
    Assume it does not hold, then choose some vertex $z$ with smallest height in $G'$ for which the statement fails.
    The choice of $z$ with minimal height implies that $u(\{\xi\})$ and $u(X)$ agree on $N^+_{G'}(z)$.

    If $|z|< n$, then $z\in u(\{\xi\})\iff z\in u\iff z\in u(X)$, by clause (1) of the definitions of $G(\{\xi\})$, $u(\{\xi\})$, $u(X)$, contradicting the assumption.

    If $|z|=n$, then $z\in R_{G(\{\xi\})}(Q) \cap \Sigma^n$ and by the assumption of this Claim~\ref{claim:freezing-disjconp-simple} and Claim~\ref{claim:welldefined-disjconp}(i), we have that $z\in u(\{\xi\})\iff z\in\{\xi\} \iff z\in X\iff z\in u(X)$. This contradicts our assumption.

    Thus, we obtain $|z|>n$. If $z$ is defined by clause (5) in $u(\{\xi\})$, then it is also defined by clause (5) in $u(X)$. We have that $z\not\in u(\{\xi\}), u(X)$, contradicting the assumption.

    If $z$ is defined by clause (4) in $u(\{\xi\})$, then it is an incorrect code word. We obtain a contradiction with the assertion that $G(\{\xi\})$ does not contain any incorrect code words, by Claim~\ref{claim:welldefined-disjconp}(ii).

    Hence otherwise, $z$ is defined by clause (3) in $u(\{\xi\})$ and a correct code word in $G(\{\xi\})$.
    Thus we know $z=c(a,b,x)$ for suitable $a,b,x$ with $0<t_s(\tau^1_{a,b})\leq z$, and at least one of $M^{{u(\{\xi\})}}_{a}(x)$ or $M_{b}^{{u(\{\xi\})}}(x)$ accepts.
    Now, either
    \begin{gather*} z\in u(\{\xi\}) \overset{\mathclap{\text{\ref{claim:welldefined-disjconp}(v)}}}{\implies} M_a^{{u(\{\xi\})}}(x) \text{ acc.} \overset{\mathclap{\text{\ref{obs:partialoracles}(iii)}}}{\implies} M_a^{{u(X)}}(x) \text{ acc.} \overset{\mathclap{\text{\ref{claim:welldefined-disjconp}(iii)}}}{\implies} z\in u(X),\end{gather*}
    where the first implication holds by Claim~\ref{claim:welldefined-disjconp}(v), and the second implication holds by Observation~\ref{obs:partialoracles}(iii) combined with the fact that $u(\{\xi\})$ and $u(X)$ agree on $N^+_{G'}(z)$, which by definition in clause (3) is precisely the set of oracle queries of $M_a^{{u(\{\xi\})}}(x)$.
    The last implication is a consequence of Claim~\ref{claim:welldefined-disjconp}(iii), having that $X$ is V8-compatible.

    Or, symmetric,
    \begin{gather*} z\not\in u(\{\xi\}) \overset{\mathclap{\text{\ref{claim:welldefined-disjconp}(v)}}}{\implies} M_b^{{u(\{\xi\})}}(x) \text{ acc.} \overset{\mathclap{\text{\ref{obs:partialoracles}(iii)}}}{\implies} M_b^{{u(X)}}(x) \text{ acc.} \overset{\mathclap{\text{\ref{claim:welldefined-disjconp}(iii)}}}{\implies} z\not \in u(X). \end{gather*}
    In both cases, we obtain a contradiction.
    We conclude that the above statement holds for all $z\in R_{G(\{\xi\})}(Q)$, and the claim follows.
\end{proof}
We want to generalize the statement from Claim~\ref{claim:frozenwords-len-simple} to the case $X=\emptyset$. Here, $X$ is V8-incompatible. For this, we first have to prove that incorrect code words do not add too many edges to $G(\emptyset)$ in clause (4).
\begin{claim}\label{claim:frozenwords-npconp}
    Let $z=c(a,b,x)$ be an incorrect code word in $G(\emptyset)$.
    Then either $Q_\mathrm{in}^z\leq |z|/n$ or $Q_\mathrm{out}^z\leq |z|/n$. 
\end{claim}
\begin{proof}
    Assume without loss of generality that $|Q^z_\mathrm{in}|\geq |Q^z_\mathrm{out}|$.
    Aiming for a contradiction, we assume that
    \begin{equation*} |z|/n < |Q^z_\mathrm{out}| \leq 2^n/2 \leq |Q^z_\mathrm{in}|. \end{equation*}
    We will find $\alpha,\beta\in\Sigma^n$ such that (a) $M_a^{{u(\{\alpha,\beta\})}}(x)$ and $M_b^{{u(\{\alpha,\beta\})}}(x)$ both accept, and (b) that $\alpha$ and $\beta$ have different parity.
    This implies that $\{\alpha,\beta\}$ is V8-compatible.
    Noting that $u(\{\alpha,\beta\})$ is $t_s$-valid by Claim~\ref{claim:welldefined-disjconp}(ii), we have $0<t_s(\tau^1_{a,b})\leq z<|u(\{\alpha,\beta\})|$, and we obtain a contradiction by Lemma~\ref{lemma:npconp-tasks}.

    By definition of $Q^z_\mathrm{in}$, $Q^z_\mathrm{out}$:
    \begin{itemize}[noitemsep]
        \item Whenever $\alpha\in Q^z_\mathrm{in}$, then $M_a^{{u(\{\alpha\})}}(x)$ accepts (since $\alpha\in Q^z_\mathrm{in}$ implies $z\in u(\{\alpha\})$ implies accepting $M_a^{{u(\{\alpha\})}}(x)$ by Claim~\ref{claim:welldefined-disjconp}(iii)).
        \item Symmetric, whenever $\beta\in Q^z_\mathrm{out}$, then $M_b^{{u(\{\beta\})}}(x)$ accepts.
    \end{itemize}
    Let $Q_\alpha$ (resp., $Q_\beta$) be the oracle queries of the leftmost accepting path, and let $Q^+_\alpha\coloneqq R_{G(\{\alpha\})}(Q_\alpha)$ be $Q_\alpha$'s closure under successor in $G(\{\alpha\})$. Define $Q^+_\beta$ respectively.
    By Claim~\ref{claim:codewords}(iii), $\ell(Q_\alpha), \ell(Q_\beta) \leq |z|/8$.
    By Claim~\ref{claim:frozenwords-len-simple}, $\ell(Q^+_\alpha), \ell(Q^+_\beta) \leq |z|/4$.

    In order to freeze the accepting paths, we want to freeze their respective oracle queries.
    We determine some $\alpha\in Q_\mathrm{in}^z$ and $\beta\in Q_\mathrm{out}^z$ of different parity such that $u(\{\alpha,\beta\})$ retains the respective accepting paths.
    To this end, we choose $\alpha$ and $\beta$ such that 
    \begin{enumerate}[label=(\alph*)]
        \item $\alpha\not\in Q^+_\beta$, $\beta\not\in Q^+_\alpha$, ensuring that the respective paths are retained, and
        \item $\alpha$ and $\beta$ have different parity, ensuring that $\{\alpha, \beta\}$ is V8-compatible, thus $u(\{\alpha,\beta\})$ is $t_s$-valid.
    \end{enumerate}
    Such pair $\alpha\in Q_\mathrm{in}^z$ and $\beta\in Q_\mathrm{out}^z$ satisfying the above properties exists; we postpone the proof to the end. For now, fix these words $\alpha$ and $\beta$.

    Note that $\{\alpha\}$ and $\{\alpha,\beta\}$ agree on $Q^+_\alpha\cap\Sigma^n$, because $\beta\not\in Q^+_\alpha\cap\Sigma^n$ by (a).
    Invoking Claim~\ref{claim:freezing-disjconp-simple}, we obtain that $u(\{\alpha\})$ and $u(\{\alpha,\beta\})$ agree on $Q_\alpha$.
    With Observation~\ref{obs:partialoracles}(iii), we obtain that $M_a^{{u(\{\alpha,\beta \})}}(x)$ accepts.

    Analogously, we obtain that also $M_b^{{u(\{\alpha,\beta \})}}(x)$ accepts.
    This is a contradiction, as already argued.

    \begin{center}
        $\ast$
    \end{center}

    It remains to prove that there exists a pair of $\alpha\in Q_\mathrm{in}^z$ and $\beta\in Q_\mathrm{out}^z$ satisfying above properties (a) and (b).
    Without loss,
    \[ |Q^z_\mathrm{out}\cap\Sigma^{n-1}0| \leq |Q^z_\mathrm{out}|/2 \leq |Q^z_\mathrm{out}\cap\Sigma^{n-1}1|. \]
    We want to show that $|Q^z_\mathrm{in}\cap\Sigma^{n-1}0|$ and $|Q_\mathrm{out}^z\cap\Sigma^{n-1}1|$ both have more than $|z|/(2n)$ elements.
    With the initial assumption that $|Q^z_\mathrm{out}|>|z|/n$, we immediately obtain that $|Q^z_\mathrm{out}\cap\Sigma^{n-1}1|> |z|/(2n)$.
    Furthermore, we know that $|Q^z_\mathrm{out}|\leq 2^n/2$ by the initial assumption at the beginning of this Claim~\ref{claim:frozenwords-npconp}. We thus have
    \[ |Q^z_\mathrm{out}\cap\Sigma^{n-1}0| \leq |Q^z_\mathrm{out}|/2 \leq 2^{n}/4. \]
    This implies that
    \begin{align*}
        |Q_\mathrm{in}\cap\Sigma^{n-1}0| &= |\Sigma^{n-1}0| - |Q^z_\mathrm{out}\cap\Sigma^{n-1}0| \geq 2^n/2-2^n/4 = 2^n/4,
        \intertext{and with \eqref{eq:expbound-disjconp} we have }
        &> \gamma(n)/(2n) \geq |z|/(2n).
    \end{align*}
    When we now pick $\alpha$ from $Q^z_\mathrm{in}\cap\Sigma^{n-1}0$ and $\beta$ from $Q_\mathrm{out}\cap\Sigma^{n-1}1$, we already satisfy property~(b).
	
    For the other property~(a), we set up a directed bipartite ``spoilage'' graph $S$ with left vertex part $A=Q_\mathrm{in}^z\cap\Sigma^{n-1}0$ and right vertex part $B=Q_\mathrm{out}^z\cap\Sigma^{n-1}1$.
    Note that by above argument, we have that $|A|,|B|>|z|/(2n)$.
    Define the edge set as follows: for every pair $\alpha\in A, \beta\in B$,
    \begin{gather*}
        (\alpha, \beta) \in E(S) \iff \beta\in Q^+_\alpha\cap B\text{ ($\alpha$ spoils $\beta$)},\\
        (\beta, \alpha) \in E(S) \iff \alpha\in Q^+_\beta\cap A\text{ ($\beta$ spoils $\alpha$)}
    \end{gather*}
	
    Since $\ell(Q_{\alpha}^+),\ell(Q_{\beta}^+) \leq |z|/4$, and $A,B \subseteq\Sigma^n$, the out-degree of all vertices is bounded by $\leq |z|/(4n)$. By definition, both vertex parts have $> |z|/(2n)= 2\cdot |z|/(4n)$ vertices. Conditions of combinatorial Lemma~\ref{lemma:bipartite} apply and there exist $\alpha\in Q_\mathrm{in}^z\cap\Sigma^{n-1}0, \beta\in Q_\mathrm{out}^z\cap\Sigma^{n-1}1$ such that $(\alpha,\beta)\not\in E(S)$, and $(\beta,\alpha)\not\in E(S)$. Hence, (a) $\alpha$ ends with $0$, $\beta$ ends with $1$, and (b) $\alpha\not\in Q^+_\beta$, $\beta\not\in Q^+_\alpha$, as desired.
\end{proof}
The next claim generalizes Claim~\ref{claim:frozenwords-len-simple} as desired.
\begin{claim}\label{claim:frozenwords-len-disjconp}
    Let $Q\subseteq\Sigma^{\leq \gamma(n)}$.
    Then $\ell(R_{G(\emptyset)}(Q)) \leq 2\cdot \ell(Q)$.
\end{claim}
\begin{proof}
    This follows by extending the proof of Claim~\ref{claim:frozenwords-len-simple} with the observation just made.
    Specifically, it remains to prove the induction statement for an incorrect code word $z\in V(G')$ with $\mathrm{outdeg}_{G'}(z)>0$. That is, defined from clause (4).
    Again, let $a,b,x$ be suitable such that $z=c(a,b,x)$.

    Without loss of generality, $|Q_\mathrm{out}^z|\leq |Q_\mathrm{in}^z|$,
    and $N^+_{G'}(z) = Q_\mathrm{out}^z$ by definition in clause (4).
    Then the previous Claim~\ref{claim:frozenwords-npconp} says that  $|N^+_{G'}(z)| = |Q_\mathrm{out}^z|\leq |z|/n$.
    Hence $\ell(N^+_{G'}(z))\leq |z|$, since $Q_\mathrm{out}^z\subseteq \Sigma^n$.

    Moreover, all $z'\in Q_\mathrm{out}^z$ have height 0 in $G'$. It follows that
    \begin{equation*} \ell(R_{G'}(z)) = |z| + \ell(N^+_{G'}(z))  \leq 2|z|.\qedhere \end{equation*}
\end{proof}
After proving several properties of $u(X)$, we now can look at the computations of interest, namely $M_i(F_r(0^n))$ and $M_j(F_r(0^n))$, relative to $u(X)$. The remaining proof works as follows: First, we proof that one of the stated computations accepts relative to $u(\emptyset)$. We can then look at an accepting path and put the at most polynomially many oracle queries into a set $Q$. Using our previous claims and combinatorial arguments, we finally show that we can extend $u(\emptyset)$ to $u(\{\xi\})$ for some $\xi \in \Sigma ^n$ having a desired suffix to achieve~\ref{thm:disjconp}(i) or~\ref{thm:disjconp}(ii), while freezing the words in $Q$ and therefore freezing the accepting path.
\begin{claim}
    One of $M_i(F_r(0^n))$ or $M_j(F_r(0^n))$ definitely accepts relative to $u(\emptyset)$.
\end{claim}
\begin{proof}
    Note that the runtime of both computations is bounded by $\leq \gamma(n)$ and $u(\emptyset)$ is defined for all words $\leq \gamma(n)$, thus both computations are definite.
    Now assume that both reject. Remember that $\hat{s}$ is the stage that treated $\tau^4_{i,j}$.
    Let $t'\coloneqq t_{\hat{s}-1}\cup\{\tau^4_{i,j}\mapsto 0\}$.

    We show that $u(\emptyset)$ is $t'$-valid and $u(\emptyset)\sqsupsetneq w_{\hat{s}-1}$.
    The last assertion is immediate, since $w_{\hat{s}-1} \sqsubseteq w_{s-1} \sqsubseteq u \sqsubsetneq u(\emptyset)$.
    Also note that $u(\emptyset)$ is $t_{\hat{s}-1}$-valid:
    We have $\hat{s}-1<\hat{s}$, hence conditions of Claim~\ref{claim:welldefined-disjconp}(iv) apply.

    Hence, for $t'$-validity, only V7 is at risk, but by assumption, with $x=F_r^{{u(\emptyset)}}(0^n)$, both $M_a^{{u(\emptyset)}}(x)$ and $M_b^{{u(\emptyset)}}(x)$ definitely reject.

    As $u(\emptyset)$ is a possible $t'$-valid extension of $w_{\hat{s}-1}$ in stage $\hat{s}$, we obtain that the treatment of task $\tau^4_{i,j}$ would define $t_{\hat{s}}=t'$.
    But then $t_s(\tau^4_{i,j})=t'(\tau^4_{i,j})=0$, contradicting the hypothesis of this Theorem~\ref{thm:disjconp}.
\end{proof}

For now, assume without loss of generality that $M_i^{{u(\emptyset)}}(F_r^{{u(\emptyset)}}(0^n))$ definitely accepts. 
Define $Q$ as the set of oracle queries on the leftmost accepting path,
and define $Q^+\coloneqq R_{G(\emptyset)}(Q)$.
We will find a suitable word $\xi\in\Sigma^{n-1}0$ such that $M_i(F_r(0^n))$ accepts relative to $u(\{\xi\})$.
We show that this is an appropriate choice for the desired oracle $w_s$, which satisfies statement (i) of this Theorem~\ref{thm:disjconp}.

We have $\ell(Q)\leq \gamma(n)$ since $\gamma(n)$ is an upper bound for the computation time of $M_i^{{u(\emptyset)}}(F_r^{{u(\emptyset)}}(0^n))$.
With Claim~\ref{claim:frozenwords-len-disjconp} and \eqref{eq:expbound-disjconp}, $|Q^+|\leq 2\gamma(n)<2^{n-1}$.
Hence there exists $\xi\in \Sigma^{n-1}0$ with $\xi\not\in Q^+$.
Fix such $\xi$ for the rest of the proof.

\begin{claim}\label{claim:freezing-disjconp}
    The oracles $u(\emptyset)$ and $u(\{\xi\})$ agree on $Q^+$, and thus, on $Q$.
    With Observation~\ref{obs:partialoracles}(ii) and (iii), this means that $M_i(F_r(0^n))$ definitely accepts relative to $u(\{\xi\})$.
\end{claim}
\begin{proof}
    This proof goes nearly identical to the one of Claim~\ref{claim:freezing-disjconp-simple}, we only need to additionally cover clause (4).

    Remember that $R_{G(\emptyset)}(Q)=Q^+$.
    Consider the graph $G'$ that is formed by taking the subgraph from $G(\emptyset)$ induced by $Q^+$.
    We show that
    \begin{equation*}  \text{For all words $z\in V(G')$, it holds that}\quad z\in u(\emptyset)\iff z\in u(\{\xi\}). \end{equation*}
    Assume it does not hold, then choose some vertex $z$ with smallest height for which the statement fails.

    If $|z|< n$, then $z\in u(\emptyset)\iff z\in u\iff z\in u(\{\xi\})$, by clause (1) of the definitions of $G(\emptyset)$, $u(\emptyset)$, $u(\{\xi\})$, contradicting the assumption.
    If $|z|=n$, then $z\in V(G')=Q^+$ implies $z\neq\xi$ by choice of $\xi\not\in Q^+$. Hence we have $z\not\in u(\emptyset), z\not\in u(\{\xi\})$ by Claim~\ref{claim:welldefined-disjconp}(ii), contradicting the assumption.

    Thus, we obtain $|z|>n$. If $z$ is defined by clause (5) in $u(\emptyset)$ then $z$ is also defined by clause (5) in $u(\{\xi\})$. We obtain $z\not\in u(\emptyset), u(\{\xi\})$, contradicting the assumption.

    If $z$ is defined by clause (3) in $u(\emptyset)$, then it is a correct code word in $G(\emptyset)$. We obtain a contradiction identical to the proof of Claim~\ref{claim:freezing-disjconp-simple}.
    Specifically, we have
    \begin{gather*} z\in u(\emptyset) \overset{\mathclap{\text{\ref{claim:welldefined-disjconp}(v)}}}{\implies} M_a^{{u(\emptyset)}}(x) \text{ acc.} \overset{\mathclap{\text{\ref{obs:partialoracles}(iii)}}}{\implies} M_a^{{u(\{\xi\})}}(x) \text{ acc.} \overset{\mathclap{\text{\ref{claim:welldefined-disjconp}(iii)}}}{\implies} z\in u(\{\xi\}), \text{and}\\
    z\not\in u(\emptyset) \overset{\mathclap{\text{\ref{claim:welldefined-disjconp}(v)}}}{\implies} M_b^{{u(\emptyset)}}(x) \text{ acc.} \overset{\mathclap{\text{\ref{obs:partialoracles}(iii)}}}{\implies} M_b^{{u(\{\xi\})}}(x) \text{ acc.}  \overset{\mathclap{\text{\ref{claim:welldefined-disjconp}(iii)}}}{\implies} z\not \in u(\{\xi\}). \end{gather*}

    Hence, $z$ must be defined by clause (4) in $u(\emptyset)$, i.e., $z$ is an incorrect code word in $G(\emptyset)$.
    Thus we know $z=c(a,b,x)$ for suitable $a,b,x$.
    We need to show that $z\in u(\emptyset)$ if and only if $z\in u(\{\xi\})$.
    
    If $z\in u(\emptyset)$, then $N^+_{G'}(z)=Q^z_\text{out}$ by definition in clause (4).
    We now show that this implies $\xi\in Q^z_\text{in}$.
    By minimality assumption, $u(\emptyset)$ and $u(\{\xi\})$ agree on $N^+_{G'}(z)=Q^z_\text{out}$.
    If now $\xi\not\in Q^z_\text{in}$, it is member of $Q^z_\text{out}=N^+_{G'}(z)$.
    This means that $u(\emptyset)$ and $u(\{\xi\})$ agree on $\{\xi\}\subseteq N^+_{G'}(z)$.
    This is a contradiction, since $\xi\not\in u(\emptyset)$ and $\xi\in u(\{\xi\})$, by Claim~\ref{claim:welldefined-disjconp}(ii).
    Hence $\xi\in \Sigma^n\setminus Q_\text{out}^z = Q^z_\text{in}$, and by definition, $z\in u(\{\xi\})$.

    The case for $z\not\in u(\emptyset)$ follows by the symmetric argument.
    We conclude that the above statement holds for all $z\in Q^+$, and the claim follows.
\end{proof}

It remains to show that $w=u(\{\xi\})$ satisfies the condition (i) stated in this Theorem~\ref{thm:disjconp}.
The first part is immediate, since $0^n \in D^{{u(\{\xi\})}}_m$, hence for any extension $v\sqsupseteq w =u(\{\xi\})$, also $0^n \in D^v_m$.
The second part follows from Claim~\ref{claim:freezing-disjconp}.

This completes the proof of Theorem~\ref{thm:disjconp}.
\end{proof}

We have now completed the proofs showing that the oracle construction can be performed as desired.
The following theorem confirms the desired properties of $O\coloneqq\bigcup_{i\in\mathbb N} w_i$.
Remember that $|w_0|<|w_1|<\dots$ is unbounded, hence for any $z$ there is a sufficiently large $s$ such that $|w_s|>z$.
Remember that $w_s$ is $t_s$-valid for all $s\in\mathbb N$.

\begin{theorem}\label{thm:result}
    Relative to $O=\bigcup_{i\in\mathbb N} w_i$, the following holds:
    \begin{enumerate}
        \item $\NPcoNP=\P$, which implies $\neg\hNPcoNP$.
        \item No pair in $\DisjNP$ is $\leqmpp$-hard for $\DisjUP$, which implies $\hDisjNP$.
        \item No language in $\UP$ is $\leqmp$-complete for $\UP$, i.e., $\hUP$.
        \item No pair in $\DisjCoNP$ is $\leqmpp$-hard for $\DisjCoUP$, which implies $\hDisjCoNP$.
    \end{enumerate}
\end{theorem}
\begin{proof}
    To (i): It suffices to show that each $L\in\mathrm{NP}^O\cap\mathrm{coNP}^O$ is decidable in $\P^O$.
    Choose $a,b$ such that $L(M^O_a)=L$ and $L(M^O_b)=\overline{L}$.
    Let $s$ be the stage that treated $\tau^1_{a,b}$.
    If $t_{s}(\tau^1_{a,b})=0$, then V1 posits that for some $x\in\Sigma^*$, relative to $w_s$, $M_a(x)$ definitely accepts if and only if $M_b(x)$ definitely accepts.
    Since $w_{s}\sqsubseteq O$, $M_a^O(x)$ accepts if and only if $M_b^O(x)$ accepts, by Observation~\ref{obs:partialoracles}(i).
    This contradicts the choice of $M_a, M_b$.

    Thus, $z\coloneqq t_s(\tau^1_{a,b})>0$.
    Consider the following union of two disjoint sets
    \begin{equation*} L' \coloneqq \{ x \mid c(a,b,x)\in O, x\geq z\} \cup \{ x\in L \mid x< z\}. \end{equation*}
    With Claim~\ref{claim:codewords}(ii), $L'\in\P^O$.
    We now show $L=L'$: Let $x\in\Sigma^*$.
    If $x< z$, then trivially $x\in L$ if and only if $x\in L'$.

    If $x\geq z$, then choose a sufficiently large $s'$ such that $|w_{s'}|>c(a,b,x)$.
    Note that $c(a,b,x)\in O$ if and only if $c(a,b,x)\in w_{s'}$.
    Also, $0<t_{s'}(\tau^1_{a,b})=z\leq x\leq c(a,b,x) < |w_{s'}|$, and V2 applies with respect to $t_{s'}$-valid $w_{s'}$.
    Specifically, the computation $M_a^{w_{s'}}(x)$ is definite.
    Thus,
    \begin{align*}
        x\in L&\iff M_a^{O}(x) \text{ acc.}  \overset{\mathclap{\text{\ref{obs:partialoracles}(i)}}}\iff M_a^{w_{s'}}(x) \text{ acc.} \overset{\mathclap{\text{V2}}}\iff c(a,b,x)\in w_{s'}\\
              &\iff c(a,b,x)\in O \iff x\in L'. 
    \end{align*}
    This implies $L\in\P^O$.
    \medskip

    To (ii): Assume otherwise that there is a pair $(A,B)\in \DisjNP^O$ that is \leqmpp-hard for $\DisjUP$.
    Choose some $i,j$ such that $L(M^O_i)=A$ and $L(M^O_j)=B$. Consider the task $\tau^2_{i,j}$, that is treated in some stage $\hat{s}$.

    If $t_{\hat{s}}(\tau^2_{i,j})=0$, then V3 states that for some $x\in\Sigma^*$, both $M_i(x)$  and $M_j(x)$ definitely accept relative to $w_{\hat{s}}$. Since $w_{\hat{s}}\sqsubseteq O$, $M_i(x)$ and $M_j(x)$ accept relative to $O$, by Observation~\ref{obs:partialoracles}(i).
    This implies $x\in A\cap B$, contradicting the assumption that $A$ and $B$ are disjoint.

    Hence we can assume that $m\coloneqq t_{\hat{s}}(\tau^2_{i,j})>0$.
    Consider the pair $(A^O_m, B^O_m)$. We show with V4 that this pair is a disjoint UP-pair.
    Assume otherwise, then contraposition of Observation~\ref{obs:witnesses}(i) asserts that for some $n$, we have $|O\cap \Sigma^n|>1$.
    Now choose some sufficiently large $s'$ such that $w_{s'}$ is defined for all words of length $n$.
    As $w_{s'}$ is $t_{s'}$-valid, V4 applies and states that $|\Sigma^n\cap w_{s'}|\leq 1$. As $w_{s'}$ is chosen to be defined for \emph{all} words of length $n$, $w_{s'}$ and $O$ agree on $\Sigma^n$, hence we obtain $|\Sigma^n\cap O|\leq 1$, contradicting the assumption.

    Hence $(A^O_m, B^O_m)\in\DisjUP^O$. This implies that $(A^O_m, B^O_m) \leqmpp (A,B)$.
    Choose $r$ such that the reduction is computed by $F^O_r$, and consider the task $\tau^2_{i,j,r}$, that is treated in some stage $s$.
    By the definition of the task, it follows that, without loss of generality, there is an $n\in\mathbb N$ and $w_s\sqsubseteq O$ such that $0^n\in A_m^O$ and $M_i^{w_s}(F_r^{w_s}(0^n))$ definitely rejects.

    Invoking Observation~\ref{obs:partialoracles}(i), we obtain that $0^n\in A_m^O$ and $F_r^O(0^n)\not\in A$. This contradicts the assumption that $F_r^O$ realizes the reduction.
    \medskip

    To (iii): This can be proven symmetrically, referring to V5, V6.
    \medskip

    To (iv): This can be proven symmetrically, referring to V7, V8, and Observation~\ref{obs:witnesses}(ii).
\end{proof}

From Theorem~\ref{thm:result} and known relativizable results, we obtain the following additional properties that hold relative to the oracle.
\begin{corollary}\label{cor:detailed_properties}
    The following holds relative to the oracle $O$ constructed in this section.
    \begin{enumerate}[midpenalty=0]
        \item $\P = \NP \cap \coNP \subsetneq \UP \subsetneq \NP$
        \item $\UP$, $\NP$, $\NE$, and $\NEE$ are not closed under complement.
        \item $\UP \not\subseteq \coNP$
        \item $\NEE \cap \TALLY \not\subseteq \coNEE$
        \item $\NPSVt \subseteq \PF$
        \item $\NPbVt \not\subseteq_c \NPSVt$
        \item $\NPkVt \not\subseteq_c \NPSVt$ for all $k \ge 2$
        \item $\NPMVt \not\subseteq_c \NPSVt$
        \item $\NPMV \not\subseteq_c \NPSV$
        \item $\TFNP \not\subseteq_c \PF$
        \item $\NP \cap \coNP$ has $\leqmp$-complete sets, i.e., $\neg \hNPcoNP$.
        \item $\UP$ has no $\leqmp$-complete sets, i.e., $\hUP$.
        \item $\DisjNP$ has no $\leqmpp$-complete pairs, i.e., $\hDisjNP$.
        \item $\DisjCoNP$ has no $\leqmpp$-complete pairs, i.e., $\hDisjCoNP$.
        \item No pair in $\DisjNP$ is $\leqmpp$-hard for $\DisjUP$.
        \item No pair in $\DisjCoNP$ is $\leqmpp$-hard for $\DisjCoUP$.
        \item There are no p-optimal proof systems for $\TAUT$, i.e., $\hCON$.
        \item There are no optimal proof systems for $\TAUT$.
        \item There are no p-optimal proof systems for $\SAT$, i.e., $\hSAT$.
        \item $\TFNP$ has no $\leqmp$-complete problems, i.e., $\hTFNP$.
        \item $\NPMVt$ has no $\leqmp$-complete functions.
        \item $\NP$ and $\coNP$ do not have the shrinking property.
        \item $\NP$ and $\coNP$ do not have the separation property.
        \item $\DisjNP$ and $\DisjCoNP$ contain $\P$-inseparable pairs.

    \end{enumerate}
\end{corollary}
\begin{proof}
    The corollary holds by the following arguments, which are relativizable.

    (i): Follows from Theorem~\ref{thm:result}(i) and \ref{thm:result}(iii).

    (xi)-(xvi): Follows from Theorem~\ref{thm:result}.

    (xvii)--(xviii): 
    Assume there is an optimal proof system for TAUT. Köbler, Messner, and Torán \cite[Cor.~6.1]{kmt03} show that this implies that there is a disjoint NP-pair that is $\leqmpp$-complete. This contradicts (xiii). Statement (xviii) follows immediately.

    (xx): Follows from (xiv) by Pudlák \cite[Prop.~5.6]{pud17}.

    (xxi): Follows from (xx) by Pudlák \cite[Prop.~5.10]{pud17}.

    (xix): Follows from (xxi)
    by Beyersdorff, Köbler and Messner \cite[Thm.~25]{bkm09}.

    (iv): Follows from (xviii) by Köbler, Messner, and Torán \cite[Corollary 7.1]{kmt03}.

    (ii): $\UP \neq \coUP$, since otherwise
    $\UP=\coUP=\UP\cap\coUP\subseteq \NP\cap\coNP=\P$, which contradicts (i).
    The same argument shows $\NP \neq \coNP$.
    By (iv), $\NEE \neq \coNEE$ which implies $\NE \neq \coNE$ by padding.

    (iii): $\UP \not\subseteq \coNP$, since otherwise
    $\UP \subseteq \NP \cap \coNP \subseteq \P$, which contradicts (i).

    (v): Statement (i) implies $\P=\NPcoNP$, which is equivalent to (v), as shown by Fenner et al.\ \cite[Prop.~1]{ffnr03}.

    (vi): Assume $\NPbVt \subseteq_c \NPSVt$. By (v), $\NPbVt \subseteq_c \PF$.
    By Fenner et al.\ \cite[Thm.~4]{ffnr03},
    this implies that all disjoint $\coNP$-pairs are $\P$-separable,
    which contradicts (xiv).

    (vii): By (vi), $\NPbVt \not\subseteq_c \PF$.
    Fenner et al.\ \cite[Thm.~14]{ffnr03} show that this implies
    $\NPkVt \not\subseteq_c \PF$ for all $k \ge 2$.
    From (v) it follows $\NPkVt \not\subseteq_c \NPSVt$ for all $k \ge 2$.

    (viii): Follows from (vi).

    (ix): Follows from (viii).

    (x): Assume $\TFNP \subseteq_c \PF$.
    By Fenner et al.\ \cite[Prop.~7, Thm.~2]{ffnr03},
    this implies $\NPMV_t \subseteq_c \PF$, which contradicts (viii).

    (xxii): Follows from (xiii) and (ii), because of the following results.
    Glaßer, Reitwießner, and Selivanov \cite[Thm.~3.7]{grs11} show that
    if there is no $\leqmpp$-complete pair for $\DisjNP$, then $\NP$ does not have the shrinking property.
    In the same publication \cite[Thm.~3.4]{grs11}, they show that
    $\coNP$ has the shrinking property if and only if $\NP = \coNP$.

    (xxiii): Glaßer, Reitwießner, and Selivanov \cite[Thm.~3.9, Thm.~3.10]{grs11}
    show that (iii) (resp., (vi)) implies that $\NP$ (resp., $\coNP$)
    has not the separation property.

    (xxiv): Follows from (xiii) and (xiv).
\end{proof}

\printbibliography

\end{document}